\documentclass[11pt]{extarticle}

\usepackage[a4paper]{geometry}
\usepackage{amsmath,amssymb,amsthm,mathtools}
\usepackage{xparse}
\usepackage{dsfont}
\usepackage{xcolor}
\usepackage{lmodern}
\usepackage{diffcoeff}
\usepackage{commath}
\usepackage{graphicx}
\usepackage{pgfplots} 
\pgfplotsset{compat=1.14, every tick label/.append style={font=\small}}
\usepackage{authblk}

\numberwithin{equation}{section}
\newcounter{dummy} \numberwithin{dummy}{section}
\newtheorem{theorem}[dummy]{Theorem}
\newtheorem{lemma}[dummy]{Lemma}
\newtheorem{proposition}[dummy]{Proposition}

\newtheorem{definition}[dummy]{Definition}

\newtheorem{remark}[dummy]{Remark}

\newcounter{assum}

\binoppenalty=\maxdimen
\relpenalty=\maxdimen

\DeclareMathOperator{\supp}{supp}

\newcommand{\F}{\mathcal{F}}

\newcommand{\R}{\mathbb{R}}

\renewcommand{\c}{\mathfrak{c}}

\newcommand{\iX}{\mathcal{X}}

\newcommand{\N}{\mathbb{N}}
\newcommand{\Z}{\mathbb{Z}}
\renewcommand{\L}{\mathfrak{L}}
\newcommand{\A}{\mathcal{A}}
\newcommand{\eps}{\varepsilon}
\newcommand{\Err}{\mathcal{E}}
\newcommand{\la}{\lambda}
\newcommand{\La}{\Lambda}

\newcommand{\ka}{\varkappa}

\newcommand{\M}{\mathcal{M}_1}
\newcommand{\D}{\mathfrak{D}}

\newcommand{\1}{\mathds{1}}

\newcommand{\<}{\langle}
\renewcommand{\>}{\rangle}
\renewcommand{\*}{\!\cdot\!}
\renewcommand{\#}{\star}

\let\copyint\int
\RenewDocumentCommand \int {o o}
{ \IfNoValueTF {#2} { \IfNoValueTF {#1} { \copyint } { \copyint\limits_{#1} } }	{ \copyint\limits_{#1}^{#2} } }

\title{Hydrodynamics of a particle model in contact with stochastic reservoirs.}
\date{}
\author[1]{Pasha Tkachov \thanks{pasha.tkachov@gssi.it}}
\affil[1]{GSSI, Mathematics, Viale Francesco Crispi 7, 67100 L'Aquila, Italy}

\begin{document}

\maketitle
\vspace{-10mm}
\begin{abstract}
	We consider an exclusion process with finite-range interactions in the microscopic interval $[0,N]$. 
	The process is coupled with the simple symmetric exclusion processes in the intervals $[-N,-1]$ and $[N+1,2N]$, which simulate reservoirs. 
	We show that an average of the empirical densities of the processes speeded up by the factor $N^2$ converge to solutions of parabolic partial differential equations inside $[-N,-1]$, $[0,N]$ and $[N+1,2N]$, which correspond to the macroscopic intervals $(-1,0)$, $(0,1)$ and $(1,2)$. 
	Since the total number of particles is preserved by the evolution, we obtain the Neumann boundary conditions on the external boundaries $u=-1$, $u=2$ of the reservoirs.
	Finally, a system of Neumann and Dirichlet boundary conditions is derived at the interior boundaries $u=0$, $u=1$ of the reservoirs.
\end{abstract}

\section{Introduction}
Our aim is to study the hydrodynamic limit of a discrete lattice gas model on a finite interval.
The molecules of a gas migrate between adjacent sites in $[0,N]$ randomly according to a spatially homogeneous exclusion process with finite-range interaction.
The interval $[0,N]$ lies between two reservoirs modeled by the intervals $[-N,-1]$ and $[N+1,2N]$.
The molecules may move between the interval $[0,N]$ and the reservoirs, in which they migrate according to the symmetric simple exclusion process.
We assume that the particles do not escape the external boundaries of the reservoirs. 
Hence, the total number of the molecules is preserved in the interval $[-N,2N]$.


\vspace{5mm}
\begin{tikzpicture}[trim axis left] 
	\begin{axis}[width=1.04*\textwidth,
							 height=0.4*\textwidth,
							hide y axis=true,
							axis x line=middle,
							hide obscured x ticks=false,
							scatter/classes={
								b={mark=square*,blue},
								a={mark=triangle*,red},
								c={mark=*,draw=black}
							},
							xtick={0,6,7,8,14,15,16,22},
							xticklabels={$-N$,$-1$,$0$,$1$,$N{-}1$,$N$,$N{+}1$,$2N$},
						]
	\addplot[blue,scatter, scatter src=explicit symbolic,] 
			coordinates {(0,0) [b] 
									 (1,0) [b] 
								   (2,0) [b]  
								   (3,0) [b]  
								   (4,0) [b]  
								   (5,0) [b]  
								   (6,0) [c]  
									};

	\draw [thick,decoration={brace,mirror,raise=15pt},decorate,blue] 
		(axis cs:0,0) -- node[below=20pt] {\small{\color{black}Simple exclusion process}} (axis cs:6,0);

	\addplot[black, thick, scatter, scatter src=explicit symbolic,] 
			coordinates{(6,0) [b]  
									(7,0) [c]  
								  (8,0) [a]  
								 };

	\addplot[red,scatter, scatter src=explicit symbolic,] 
			coordinates{(8,0) [a]
								  (9,0) [a]  
								  (10,0) [a] 
								  (11,0) [a] 
								  (12,0) [a] 
								  (13,0) [a] 
								  (14,0) [a] 
								 };

	\draw [thick,decoration={brace,raise=5pt},decorate,red] 
		(axis cs:8,0) -- node[above=10pt] {\small{\color{black}Exclusion process with finite-range interactions}} (axis cs:14,0);

	\addplot[black, thick, scatter, scatter src=explicit symbolic,] 
			coordinates{(14,0) [a]  
								  (15,0) [c] 
								  (16,0) [b] 
								};

	\addplot[blue,scatter, scatter src=explicit symbolic,] 
			coordinates{(16,0) [b] 
								  (17,0) [b] 
								  (18,0) [b] 
								  (19,0) [b] 
								  (20,0) [b] 
								  (21,0) [b] 
								  (22,0) [b] 
								}; 

	\draw [thick,decoration={brace,mirror,raise=15pt},decorate,blue] 
		(axis cs:16,0) -- node[below=20pt] {\small{\color{black}Simple exclusion process}} (axis cs:22,0);
\end{axis}
\end{tikzpicture}
\vspace{5mm}

Our goal is to find a limit of the empirical density of the process $\{\eta^N(t)\}_{t\geq0}$ speeded up by the factor $N^2$,
\begin{equation}\label{eq:empirical_density}
	\pi^N(t) := \pi^N(\dif u,t) := \frac{1}{N} \sum_{x=-N}^{2N} \eta_x^N(t) \delta_{\frac{x}{N}}(\dif u).
\end{equation}
The main result of the article is Theorem~\ref{thm:main_theorem} below, which states that an average of the empirical measure $\pi^N(t)$ converges to a measure, which density is a weak solution (see Definition~\ref{def:very_weak_solution} and \ref{def:weak_Neumann}) of the parabolic partial differential equation for $t>0$, 
\begin{gather}
	\label{eq:PDE} 	\diffp[1] \rho t(u,t) = 
			\left\{ \begin{aligned} 
				\diffp[2]{\rho}{u}&(u,t), \quad & &u\in(-1,0)\cup(1,2),\\ 
				\diffp[2]{}{u} \Phi&(\rho(u,t)), \quad & &u\in(0,1),
			\end{aligned} \right.\\
	\diffp[1]{\rho}{u}(-1_+,t) = \diffp[1]{\rho}{u}(2_-,t) = 0, \label{eq:Neum} \\
	\diffp[1]{\rho}{u}(0_-,t) = \diffp[1]{}{u} \Phi(\rho(0_+,t)),\qquad \diffp[1]{}{u} \Phi(\rho(1_-,t)) = \diffp[1]{\rho}{u}(1_+,t), \label{eq:Neumann_boundary} \\
	\la^-_{\rho(0_-,t)} = \la^{+}_{\rho(0_+,t)}, \quad \la^{+}_{\rho(1_-,t)} = \la^-_{\rho(1_+,t)}, \label{eq:Dirichlet_boundary}
\end{gather}
where $\Phi:[0,1]\to[0,\infty)$ in \eqref{eq:PDE} is non-linear in general.
Since the total number of molecules is preserved in time, then at $u=-1$ and $u=2$ we obtain the Neumann boundary conditions \eqref{eq:Neum}.
Continuity of the flow of the particles moving through $u=\frac{x}{N}=0$ and $u=\frac{x}{N}=1$ implies \eqref{eq:Neumann_boundary}. 
The continuity of the chemical potential corresponding to the local density of the molecules $\rho$ gives \eqref{eq:Dirichlet_boundary},
where the functions  $\rho\to \la^\pm_\rho:(0,1)\to\R$ are defined by the so-called equivalence of ensembles at the level of potentials (see Subsection~\ref{sec:equivalence_of_ensembles} for the rigorous definitions).
Note that in general the density $\rho(\cdot,t)$ is discontinuous at $u=0$ and $u=1$ by \eqref{eq:Dirichlet_boundary}.

There is no surprise, that in the macroscopic interval $(0,1)$, where the process evolves as a finite-range exclusion process, its empirical density converges to a solution of a parabolic equation with the nonlinear diffusion,  while in $(-1,0)$ and $(1,2)$, where the process evolves as a symmetric simple exclusion process, its empirical density converges to the (linear) heat equation \cite{ELS1990,ELS1991,FHU1991}.
The main novelty of the article are the boundary conditions \eqref{eq:Neumann_boundary} and \eqref{eq:Dirichlet_boundary} at $u=0$ and $u=1$, which glue the equations in \eqref{eq:PDE} together.

The classical work on the hydrodynamic limits is \cite{GPV1988}, where the nonlinear diffusion equation on the torus was derived from the diffusion process on the periodic lattice.
Similar result for the exclusion process with finite-range interactions was obtained in \cite{FHU1991}.
In both papers boundary conditions did not appear due to periodicity of the underlying space.
A model similar to ours was studied in \cite{ELS1990,ELS1991}, where instead of the microscopic intervals $[-N,-1]$, $[N+1,2N]$, the reservoirs were modeled fixing the chemical potentials at $\frac{x}{N}=0$ and $\frac{x}{N}=1$.
As a result, the nonlinear diffusion equation with Dirichlet boundary conditions was derived.
The boundary conditions of the form \eqref{eq:Neumann_boundary} and \eqref{eq:Dirichlet_boundary} were obtained in \cite{LOV1997}, where the hydrodynamic behavior of the zero range process with an asymmetry at the origin was studied.
In \cite{MPTV2011} the symmetric simple exclusion process was considered in the interval $[0,N]$, and the reservoirs were presented in form of a birth-death process, which was attached to the boundaries of the interval $[0,N]$.
The authors proved that the limiting density satisfies the linear heat equation with implicit Dirichlet boundary conditions.
In the preprints \cite{Ngu2018i,Ngu2018ii} the symmetric simple exclusion process interacts on the boundaries of the interval $[0,N]$ with reservoirs of various lengths.
The author considered a pairwise interaction of the particles at $\frac{x}{N}=0$ and $\frac{x}{N}=1$ with all particles in the adjacent reservoirs, deriving the linear heat equation with various Neumann-type boundary conditions.
The model is different of our case, since we allow only finite range interactions with the particles in the reservoirs.
As a result, in contrast with \cite{Ngu2018i,Ngu2018ii}, we can not neglect the evolution inside the reservoirs.

Our derivation of the parabolic equation \eqref{eq:PDE} is based on the papers \cite{ELS1990,ELS1991,FHU1991}. Although, for the convenience of the reader we attempted to repeat the relevant material from the papers in order to make our exposition self-contained, some details were omitted, as it required us to repeat the articles almost completely.

The paper is organized as follows: In Section~\ref{sec:rigorous_results} our main result is rigorously stated (see Theorem~\ref{thm:main_theorem}).
In Section~\ref{sec:notations} we set up notations and terminology needed for the proofs.
Section~\ref{sec:scaling_limit} establishes the proof of the main result.
The rest of the paper is devoted to the justification of several technical results we used in Section~\ref{sec:scaling_limit} including the so-called replacement lemma.

\textbf{Acknowledgments.} The author is greatly indebted to Professor Errico Presutti and Professor Anna De Masi for suggesting the problem and for many stimulating conversations.

\section{Rigorous results}\label{sec:rigorous_results}
For $x\in\{-N,\cdots,2N\}$, $N\in\N$, we consider a scaled lattice with sites $\frac{x}{N}$ on the macroscopic interval $[-1,2]$.
Thus, the adjacent sites are $\frac{x}{N}$, $\frac{x\pm1}{N}$.
There may be at most one molecule $\eta_x$ per site $\frac{x}{N}$.
Therefore, the molecules configuration is identified with a point $\eta$ in $\{0,1\}^{\{-N,\dots,2N\}}$.
The generator of the exclusion process $\{\eta^N_t\}_{t\geq0}$ is given by  
\begin{equation}\label{eq:generator}
	N^2 L_N = N^2\sum\limits_{x=-N}^{2N-1} L_{x,x+1}, \qquad (L_{x,x+1} g)(\eta) = c_{x,x+1}(\eta)(g(\eta^{x,x+1}) - g(\eta)),
\end{equation}
where $\eta^{x,x+1}$ is obtained from $\eta$ exchanging values of $\eta_x$ and $\eta_{x+1}$, and $g$ is a real-valued function on the configuration space.
The factor $N^2$ represents the scaling of time $t\to N^2 t$.

We assume that the exchange rates $c_{x,x+1}$ between sites $\frac{x}{N}$ and $\frac{x+1}{N}$ in the interval $(0,1)$ have finite-range. 
We will study the interactions similar to \cite[Example~2]{FHU1991}, namely for $x\in\{1,\dots,N-2\}$, $\theta,\alpha,\beta\in\R$, $\theta+\alpha>0$, $\theta+\beta>0$, $\theta+\alpha+\beta>0$, 
\begin{equation}\label{eq:nonlocal_exchange_rates}
	c_{x,x+1}(\eta) = \eta_x(1-\eta_{x+1})(\theta + \alpha \eta_{x-1} + \beta \eta_{x+2}) + \eta_{x+1}(1-\eta_x)(\theta +\alpha \eta_{x+2} + \beta \eta_{x-1}).
\end{equation}
For $\frac{x}{N}\in [-1,0)\cup(1,2]$ we consider the simple symmetric exclusion process: 
\begin{equation}\label{eq:simple_symmetric_exchange_rates}
	c_{x,x+1}(\eta) = \eta_x(1-\eta_{x+1}) + \eta_{x+1}(1-\eta_x).
\end{equation}
Note, that the exchange rates defined by \eqref{eq:nonlocal_exchange_rates} and \eqref{eq:simple_symmetric_exchange_rates} satisfy the gradient condition. 
Namely, for $x\in\{-N,\dots,-2\}\cup\{N+1,\dots,2N-1\}$, the exchange rates in \eqref{eq:simple_symmetric_exchange_rates} satisfy for $\Pi(\eta):=\eta_0$,
\begin{equation}\label{eq:DB_SSEP}
		\tau_{x+1}\Pi(\eta) - \tau_x\Pi(\eta) = c_{x,x+1}(\eta)(\eta_{x+1}-\eta_x),
\end{equation}
where $\tau$ denotes the shift operator: $(\tau_x \eta)_y := \eta_{y+x}$, $\tau_xh(\eta):= h(\tau_x\eta)$.
And, for $c_{x,x+1}$ in \eqref{eq:nonlocal_exchange_rates}, there exists a local function $h$ such that 
	\begin{equation}\label{assum:gradient_condition}
		\tau_x h(\eta) - \tau_{x+1} h(\eta) = c_{x,x+1}(\eta)(\eta_x-\eta_{x+1}), \quad x\in\{1,\dots,N-2\}. 
	\end{equation}
	One can define $h$ explicitly (for $\theta{=}1$ see \cite[(5.7)]{FHU1991}). However, we omit the explicit formula, since we will not use it relying on general properties of $h$ (cf.~Remark~\ref{rem:main_remark} below).

The exchange rates $c_{-1,0}$, $c_{0,1}$, $c_{N-1,N}$, $c_{N,N+1}$, are defined in a way which ensures the detailed balance condition:
\begin{equation}\label{assum:DB}
c_{x,x+1}(\eta) e^{\Delta_{x,x+1}H(\eta)} = c_{x,x+1}(\eta^{x,x+1}),\quad \Delta_{x,x+1} H(\eta) := H(\eta^{x,x+1})-H(\eta), 
\end{equation}
where, in our case, the Hamiltonian $H=H_{N}$ is defined as follows,
\begin{equation}\label{eq:Hamiltonian}
	H_{N}(\eta) := \sum_{x\in \{0,\dots,N-1\}} q\,\eta_x \eta_{x+1},\quad q:=\ln\frac{\theta+\alpha}{\theta+\beta}. 
\end{equation}
Hence, we put for $H=H_{N}$, $x\in\{-1,0,N-1,N\}$,
\begin{equation}\label{eq:exchange_rates_on_the_boundaries}
	c_{x,x+1}(\eta) := (\eta_x(1-\eta_{x+1}) + \eta_{x+1}(1-\eta_x)) e^{-\frac{1}{2}\Delta_{x,x+1} H(\eta)}.
\end{equation}
The detailed balance condition \eqref{assum:DB} implies that the process $\{\eta^N(t)\}_{t\geq0}$ given by the generator \eqref{eq:generator} is reversible with respect to the Gibbs measures $\nu^\la_N(\eta)\dif\eta$, $\la\in\R$, where 
\begin{equation}\label{eq:Gibbs}
	\nu^\la_{N} (\eta) := \frac{1}{Z^\la_{N}}e^{-H(\eta) + \la \sum_{x=-N}^{2N} \eta_x} 
\end{equation}
is  the (grand canonical) Gibbs distribution with the Hamiltonian $H=H_{N}$, and the normalizing constant $Z^\la_{N}$.

To ensure, that a limit of an average of the empirical density \eqref{eq:empirical_density} exists, we restrict our attention to the processes started from the initial distributions $\{\mu_0^N\}_{N\in\N}$ associated to a density profile $\rho_0:[-1,2]\to\R$, such that $\rho_0\in C((-1,0)\cup(0,1)\cup(1,2))$. Namely, for all $G\in\mathcal{D}$, 
\begin{equation}\label{assum:lim_of_init_distrib}
	\lim_{N\to\infty} \mu^N_0 \big[ \eta\in\{0,1\}^{\{-N,\dots,2N\}} \,:\, \big\vert \frac{1}{N} \sum_{x=-N}^{2N} G(\frac{x}{N}) \eta_x - \int[-1][2] G(u) \rho_0(u) \dif u \big\vert \geq \delta \big] =0,
\end{equation}
where 
\begin{equation}\label{eq:smooth_functions}
	\mathcal{D} := \{ G\in C^2([-1,2]) \,\big\vert\, G(-1)=G(0)=G(1)=G(2)=0\}.
\end{equation}


In the physical literature the parameter $\la$ in \eqref{eq:Gibbs} is called the chemical potential.
The so-called equivalence of ensembles at the level of potentials implies a one-to-one correspondence between the local densities of the molecules $0\leq\rho\leq1$ and the chemical potentials $\la=\la_\rho\in[-\infty,\infty]$ (see e.g. \cite{Geo1979}).
Let us look at $\frac{x}{N}\in(-1,0)$ and neglect for the moment the interaction at $x=0$.
Then, the simple symmetric exclusion process in $(-1,0)$ is invariant with respect to its own version of Gibbs distributions, which converges to a probability measure $\nu^{\la,-}$, as $N\to\infty$.
The equivalence of ensembles implies a one-to-one correspondence between the densities $0\leq\rho\leq1$ and chemical potentials $\la^-=\la^-_\rho\in\R$ (see Section~\ref{sec:equivalence_of_ensembles} for more details). The interval $(1,2)$ can be treated in the same way.
Similarly, if we remove the interactions at $x=0$, $\frac{x}{N}=1$ and look at the process in the interval $(0,1)$, we get its own Gibbs distributions, which converge to a probability measure $\nu^{\la,+}$. Then, the equivalence of ensembles implies the relation $\rho\leftrightarrow\la^+_\rho\in\R$. The chemical potentials $\la_\rho^\pm$ define the boundary conditions \eqref{eq:Dirichlet_boundary}.
Moreover, the function $\Phi$ in \eqref{eq:PDE} equals to the expectation of $h$ defined by \eqref{assum:DB} with respect to the Gibbs measure $\nu^{\la^+_\rho,+}$ (also see \eqref{eq:expectation_homogeneous}).
It is worth noting that we get the linear diffusion on $(-1,0)\cup(1,2)$, because the analogous expectation of the function $\Pi$ in \eqref{eq:DB_SSEP} with respect to the Gibbs measure $\nu^{\la^-_{\rho},-}$ equals $\rho$.

\begin{definition}\label{def:very_weak_solution}
	We say that $\rho$ is a weak solution to \eqref{eq:PDE} if $\rho\in L^1([-2,1]\times [0,T])$ for all $T>0$, and for all $G\in\mathcal{D}$, $t>0$,
	\begin{multline*}
		\int[-1][2] \rho(u,t) G(u) \dif u = \int[-1][2] \rho_0(u) G(u) \dif u + \int[0][t]\big( \int[0][1] \Phi(\rho(u,s)) \partial^2_{x}G(u) \dif u \\ 
			+ \big(\int[-1][0] + \int[1][2]\big) \rho(u,s) \partial^2_{x}G(u) \dif u + G'(-1) \rho(s,-1_+) + G'(0)(\Phi(\rho(s,0_+))-\rho(s,0_-))\\ 
					  + G'(1) (\rho(s,1_+)-\Phi(\rho(s,1_-))) - G'(2) \rho(s,2_-) \big)\dif s.
	\end{multline*}
\end{definition}
The definition is similar to the one of the very weak solution in \cite{Vaz2007}.
However, we do not require the test functions to depend on time.
\begin{definition}\label{def:weak_Neumann}
	For $u\in\R$ and $\rho,\tilde{\rho}\in L^1([-2,1]\times [0,T])$ for all $T>0$, we write
	\[
		\diffp[1]{\rho}{u} (u_+,t) = \diffp[1]{\tilde{\rho}}{u} (u_-,t),\quad t>0,
	\] 
	if the following limit holds 
\[
	\lim_{\substack{w\to u_+\\ v\to u_-}} \lim_{r\to0} \int[0][t] \big( \frac{\rho(w+r,s)-\rho(w,s)}{r} - \frac{\tilde{\rho}(v+r,s)-\tilde{\rho}(v,s)}{r} \big) \dif s =0, \quad t>0.
\]
\end{definition}

The closure of the set of functions $\mathcal{D}$ defined by \eqref{eq:smooth_functions} with respect to the supremum norm is the Banach space, which is equal to the set of continuous functions vanishing at the points $u\in\{-1,\,0,\,1,\,2\}$.
We identify the dual of the space with the measures on $[-1,2]$, which are finite on $(-1,0)\cup(0,1)\cup(1,2)$.
Its subsets of the non-negative measures and the probability measures we denote correspondingly by $\mathcal{M}_+$ and $\M$.
\begin{definition}\label{def:P_N_and_E_N}
	For a fixed $T>0$, we call $\mathbb{P}^N$ the probability distribution of the empirical density $\{\pi^N_t\}_{t\in[0,T]}$ defined by \eqref{eq:empirical_density} on the Skorohod space of right-continuous measure-valued trajectories $D([0,T]\to \mathcal{M}_+)$ given a random starting configuration $\eta^N(0)$ distributed by $\mu_0^N$.
We call $\mathbb{E}^N$ the expectation corresponding to $\mathbb{P}^N$.
\end{definition}

Let us consider a truncated version of the empirical density: 
\begin{equation}\label{eq:truncated_empirical_density}
	\pi^{N,k}(\dif u,t) := \frac{1}{N}\big( \sum_{x=-N+k}^{-k} + \sum_{x=k}^{N-k} + \sum_{x=N+k}^{2N-k} \big) \eta_x^N(t) \delta_{\frac{x}{N}}(\dif u),
\end{equation}
and its average 
\begin{equation}\label{eq:average_empirical_density}
	\tilde{\pi}^{N,k}(\dif u,t) := \frac{1}{k}\sum_{m=1}^{k-1} \frac{1}{m}\sum_{j=1}^{m-1} \pi^{N,j}(\dif u,t).
\end{equation}

We can now formulate our main result.
\begin{theorem}\label{thm:main_theorem}
	Let \eqref{assum:lim_of_init_distrib} hold true.
	Then, there exists 
	\[
		\rho\in L^2([0,T],H^1(K)), \quad K\Subset (-1,0)\cup(0,1)\cup(1,2),\ T>0, 
	\]
 	such that for all $G\in\mathcal{D}$, $T>0$, $t\in(0,T]$, $\delta>0$, 
	\[
		\lim_{k\to\infty}\lim_{N\to\infty} \mathbb{P}^N \big[ \int[0][t] \big\vert \int[-1][2] \tilde{\pi}^{N,k}(\dif u,s) G(u) \dif u - \int[-1][2] \rho(u,s) G(u) \dif u \big\vert \dif s \geq\delta \big] =0,
	\]
	and $\rho$ is a weak solution to \eqref{eq:PDE} with the initial condition $\rho(u,0) = \rho_0(u)$. 
	Moreover, the boundary conditions \eqref{eq:Neum} and \eqref{eq:Neumann_boundary} hold true in the sense of Definition~\ref{def:weak_Neumann},
and \eqref{eq:Dirichlet_boundary} is satisfied for almost all $t>0$.
\end{theorem}

\begin{remark}\label{rem:main_remark}
	\begin{enumerate}
		\item It is not true in general, that $\la^-_{\rho}$ equals $\la^+_{\rho}$. As a result, \eqref{eq:Dirichlet_boundary} implies that $\rho(\cdot,t)$ can be discontinuous at $u{=}0$ and $u{=}1$, in which case 
			\[
				\rho\not\in L^2([0,T],H^1([-1,2])).
			\]
		\item\label{item:generality} Theorem~\ref{thm:main_theorem} remains true if we consider more general exchange rates in \eqref{eq:nonlocal_exchange_rates} and \eqref{eq:simple_symmetric_exchange_rates}.
For example, in place of \eqref{eq:nonlocal_exchange_rates} we could consider rates satisfying the following properties:
			\begin{itemize}
				\item For $x\in\{0,\dots,N\}$, $c_{x,x+1}$ are local functions: There exists $r\geq0$ such that
					\[
						c_{x,x+1}(\eta) = c_{x,x+1}(\eta_{x-r}, \eta_{x-r+1}, \dots, \eta_{x+r+1}).
					\]
				\item For $x\in\{r,\dots, N-r-1\}$, the translation invariance holds: $c_{x,x+1}=\tau_x c_{r,r+1}$. 
				\item The exchange rates are non-degenerate: $c_{r,r+1}(\eta)>0$ for $\eta_r\neq\eta_{r+1}$. 
				\item There exists a local function $h$ such that for all $x\in\{r,\dots,N{-}r{-}1\}$ the gradient condition \eqref{assum:gradient_condition} holds. 
				\item There exists a translation-invariant Hamiltonian $H$ such that,\\ for $x\in\{r,\dots,N-r-1\}$  the detailed balance condition \eqref{assum:DB} holds.
				\item For $x\in\{-1,\dots,r-1\}\cup\{N-r,\dots,N\}$, $c_{x,x+1}$ are defined by \eqref{eq:exchange_rates_on_the_boundaries}.

			\end{itemize}
			Under this assumptions formulation of the theorem remains unchanged.
Similar assumptions could be put on rates in place of \eqref{eq:simple_symmetric_exchange_rates}, in which case nonlinear diffusion equation would be derived in the intervals $(-1,0)$ and $(1,2)$ with the corresponding amendments to the boundary conditions \eqref{eq:Neum}, \eqref{eq:Neumann_boundary}.
Despite of the possible generalizations we decided to consider the explicit rates \eqref{eq:nonlocal_exchange_rates}, \eqref{eq:simple_symmetric_exchange_rates} in order to simplify presentation of the results.
		\end{enumerate}
\end{remark}

\subsection{Examples}
By \eqref{eq:simple_symmetric_exchange_rates} and \eqref{eq:equation_rho_la} below, 
\[
	\la^-_{\rho} = \log \frac{\rho}{1-\rho}.
\]
By \eqref{eq:nonlocal_exchange_rates} and \cite[Example~2]{FHU1991}, for $\theta=1$,
\[
	\Phi(\rho) = - \frac{1}{2 r^2}  \big(\frac{\alpha}{\rho}+\frac{\beta}{1-\rho}\big) \big( 1-\sqrt{1-4r(1-\rho)\rho}\big) + \text{const}, \quad r:=\frac{\alpha-\beta}{1+\alpha}.
\]
In this case $\la^+_{\rho}$ solves the equation 
\[ 
	p'(\la^+_\rho) = \rho, \quad p(\la) = \log\big[ e^{-q+\la} +1 + \sqrt{(e^{-q+\la}+1)^2 -4(e^{-q+\la}-e^{\la})} \big] -\log2.
\]
For $\alpha=\beta$, we have
\[
	\Phi(\rho) = \theta \rho + \alpha \rho^2, \qquad \la^+_\rho = \log \frac{\rho}{1-\rho},
\]
in which case $\la^+_\rho=\la^-_\rho$, and \eqref{eq:Dirichlet_boundary} is reduced to the statement that the density $\rho$ is continuous at $u=0$ and $u=1$ for \textit{almost} all $t>0$.
Note that for $\alpha=0$ the Nash theorem \cite{Nas1958} implies the continuity at $u=0$ and $u=1$ for all $t>0$.

\subsection{On uniqueness of the limiting density}

The question of uniqueness of the limiting density $\rho$ given by Theorem~\ref{thm:main_theorem} remains open.
Although the limiting equation is of the gradient form, it seems that the discontinuities of the density due to \eqref{eq:Dirichlet_boundary} make it impossible to apply classical techniques, which require regularity of the corresponding vector field (see e.g. \cite{AGS2008}).

A more direct approach applied in \cite{FHU1991,ELS1991,LOV1997} requires us to choose an appropriate weighted Banach space, where the norm of the difference of two solutions with the same initial condition would be non-increasing.
The condition \eqref{eq:Dirichlet_boundary} suggests that such norm should be time dependent.
In this case extra terms absent in the above mentioned articles will appear. As a result a subtle choice of the norm is needed.

\subsection{On spatial dimension}

In the present paper we consider processes in a one-dimensional space.
There are several obstacles, which do not allow a straightforward extension of our results to higher dimensions. 

The first problem is existence of the phase transition, which implies non-uniqueness of the Gibbs measures in the limit $N\to\infty$ (cf.~Proposition~\ref{prop:Gibbs_measure_exist_unique} below). 
The form of the Gibbs measures in \eqref{eq:Gibbs} is similar to the ones for the Ising model with the nearest neighbor interactions.
In fact, up to the change of variables $\eta_x\to 2(\eta_x-\frac{1}{2})$ the Gibbs measures in $K=[0,N]$ (see \eqref{eq:Gibbs_grandcanonical_homogeneous}) coincide with the ones for the Ising model. It is well known that in the Ising model, for $\la=0$ and large $q$ defined in \eqref{eq:Hamiltonian}, the limiting Gibbs measures (as $N\to\infty$) are non-unique for any spatial dimension higher than one (see e.g.~\cite{FV2017}). 
In this case, the proof of the replacement lemma (Lemma~\ref{lem:replacement}) becomes more complicated. 
Possibly, one could adapt the approach in \cite{Rez1990}. 

The second problem is the entropy estimate, which we prove in Lemma~\ref{lem:entropy_estimates} below.
Our proof works only for a one-dimensional space.
For models with periodic boundary conditions, an additional spatial averaging of the probability distribution of the process allows to prove the entropy estimate for any spatial dimension (see e.g.~\cite{Fun2018,FHU1991}).
In our case, the periodicity is absent. 
Nevertheless, one could expect, that the entropy estimate holds true in any dimension, since the entropy should not grow too fast as $N\to\infty$ in the regions which are separated form the boundaries, due to the spatial homogeneity.

\section{Further notation}\label{sec:notations}
We denote the set of natural numbers $\{1,2,\dots \}$ by $\N$, the set of integers by $\Z$.
The closed ball in $\Z$ with the center at $x$ and radius $r$ will be denoted by $B_r(y)$.
We write $B_r$ for $B_r(0)$,
$K^c$ for the complement of $K\subset \Z$ in $\Z$. 
The notation $K\Subset \Z$ means that $K$ is a compact subset of $\Z$ or equivalently $K$ is a finite subset of $\Z$.
Let us also denote 
\begin{equation}\label{eq:K_r}
	K_r := \{ x+y \,\vert\, x\in K,\ y\in B_r\}. 
\end{equation}
The configuration space of particles on $K\subset \Z$ will be denoted by 
\[
	\iX_K := \{0,1\}^K, \quad \iX:=\iX_\Z, \quad \iX_N:=\iX_{[-N,2N]\cap\Z},\ N\in\N.
\]

A function $g:\iX\to\R$ is called \textit{local} if it depends on a finite number of states $\eta_x$, $x\in\Z$.

An operator $\tau$ denotes spatial shifts of configurations: $(\tau_y \eta)_x := \eta_{x+y}$. 
A superposition of a spatial shift $\tau$ and a function $g:\Z\to\R$ is denoted by $(\tau_x g)(\eta) := g(\tau_x \eta)$.
A superposition of a spatial shift $\tau$ and a measure $\mu$ is denoted by $(\tau_x\circ \mu)(\{\eta\}) := \mu(\{\tau_x \eta\})$.

Let us denote $\Pi_x(\eta):=\eta_x$, and write $\Pi$ for $\Pi_0$. For example, $\tau_x \Pi = \Pi_x$.

On $\iX$ we consider the cylindric sigma-algebras generated by the maps $\{\Pi_x\}_{x\in\Z}$ 
\[
	\F_K = \sigma\,\{\Pi_x,\ x\in K\}, \quad \F_N:= \F_{[-N,2N]\cap\Z},\quad N\in\N.
\]
The sigma-algebra generated by $\{\F_K\}_{K\Subset \Z}$ will be denoted by $\F$.

Let $r\geq0$ be such that 
\[
	c_{x,x+1}(\eta) = c_{x,x+1} (\eta_{x-r}, \eta_{x-r+1}, \dots, \eta_{x+r+1}).
\]
Although, by the definition of $c_{x,x+1}$ it would be sufficient for us to take $r=1$, in order to demonstrate generality of the approach (see Remark~\ref{rem:main_remark}, item \ref{item:generality}), we keep the general notation in the proofs.

For any $K\subset\Z$, $\eta\in\iX_{K}$, $\omega\in\iX$, we write $\eta\*\omega$ for the configuration in $\iX$:
\[
	(\eta \* \omega)_x = \left\{ 
												\begin{aligned} 
													&\eta_x, &\quad &x\in K,\\
													&\omega_x, &\quad &x\in K^c.
												\end{aligned} 
											\right.  
\]	
The number of particles and its density in $K$ will be denoted correspondingly,
\begin{equation}\label{eq:particle_number}
	N_K(\eta) := \sum_{j\in K} \eta_j, \qquad M_K(\eta):=\frac{N_K(\eta)}{|K|}.
\end{equation} 

For $K\Subset\Z$, $\eta\in\iX_K$, $\omega\in\iX$, and $q$ given by \eqref{eq:Hamiltonian}, we define Hamiltonians with free boundary condition and a boundary condition $\omega$ correspondingly by 
\begin{align}
	H_K(\eta) &:= \sum_{x\in K \,\text{and}\, x+1\in K} q \, \eta_x \eta_{x+1}, \label{eq:Hamiltonian_with_free_condition} \\ 
	H_K(\eta\*\omega) &:= \sum_{x\in K \,\text{or}\, x+1\in K} q \, (\eta\*\omega)_x (\eta\*\omega)_{x+1}.  \label{eq:Hamiltonian_with_boundary_condition}
\end{align}
Change of the energy $H$ after swap of particles is denoted by 
\[
	\Delta_{x,y}H(\eta) := H(\eta^{x,y})-H(\eta), \qquad \Delta_xH(\eta) := H(\eta^x)-H(\eta),
\]
where $\eta^x$ equals $\eta$ with $\eta_x$ replaced by $1-\eta_x$.

The grand canonical Gibbs distributions on $K\Subset\Z$ with free boundary condition and Hamiltonians $H\equiv 0$ and $H\equiv H_K$ are defined correspondingly by 
\begin{equation}\label{eq:Gibbs_grandcanonical_homogeneous}
	\nu^{\la,-}_K(\eta) := \frac{1}{Z^{\la,-}_K} e^{\la N_K(\eta)}, \quad 
	\nu^{\la,+}_K(\eta) := \frac{1}{Z^{\la,+}_K} e^{-H_K(\eta) + \la N_K(\eta)},
\end{equation}
where $Z^{\la,\pm}_K$ are the normalizing constants called grand canonical partition functions.
Note that $\nu^\la_N$ in \eqref{eq:Gibbs} is defined in the same way on $[-N,2N]\cap\Z$ with the Hamiltonian $H=H_{[0,N]}$ given by \eqref{eq:Hamiltonian}.

The probability measure $\nu_K^{\la,+}(\eta)\dif\eta$ defines the invariant measure on $\iX_K$ for the exclusion process with exchange rates given by \eqref{eq:nonlocal_exchange_rates} for $x,x+1\in K$, while $\nu^{\la,-}_K(\eta)\dif\eta$ is invariant for the simple symmetric exclusion process on $K$.
Here $\dif\eta := \prod_{x} \dif \eta_x$ denotes the counting measure.

The grand canonical Gibbs distributions on $K\Subset\Z$ with a boundary condition $\omega\in\iX$ and Hamiltonian $H\equiv H_K(\cdot\,\vert\,\omega)$ is  defined by 
\begin{equation}\label{eq:Gibbs_grandcanonical_homogeneous_boundary}
	\nu^{\la,+}_K(\eta\,\vert\,\omega) := \frac{1}{Z^{\la,+}_K(\omega)} e^{-H_K(\eta\,\vert\,\omega) + \la N_K(\eta)},
\end{equation}
where the corresponding normalizing constant $Z_K^{\la,+}(\omega)$ is the partition function with the boundary $\omega$. Note, that the analogous definition of $\nu^{\la,-}_K(\eta\,\vert\,\omega)$ coincides with the one of $\nu^{\la,-}_K(\eta)$. 

We will often abuse notations and write 
\begin{equation}\label{eq:abuse_notations_of_Gibbs}
	\nu^\la_N (\dif \eta) := \nu^\la_N(\eta)\dif\eta,\quad N\in\N; \qquad \nu^{\la,\pm}_K (\dif \eta) := \nu^{\la,\pm}_K(\eta)\dif\eta, \quad K\Subset\Z.
\end{equation}

Next, let us denote the set of the probability measures on $(\iX_K,\F_K)$ by $\M(\iX_K)$, and the ones on $(\iX,\F)$ by $\M(\iX)$.
We extend $\mu\in\M(\iX_K)$ to $\M(\iX)$ by  
\begin{equation}\label{eq:mu_N_extended}
	\mu(\eta\*\omega):=\mu(\eta),\quad \eta\in\iX_K,\ \omega\in\iX.
\end{equation}
By abuse of notations we will write $\mu(\eta)$ for  $\mu(\{\eta\})$.

For $\pi:[0,T]\to\M(\iX)$ and $G\in\mathcal{D}$ (see \eqref{eq:smooth_functions}), we put
\[
	\<\pi(t), G\> := \int[-1][2] \pi(t,\dif u)G(u). 
\]

\section{Scaling limit}\label{sec:scaling_limit}
Unless otherwise stated we assume that \eqref{assum:lim_of_init_distrib} hold true.
\subsection{Tightness of $\{\mathbb{P}^N\}_{N\in\N}$}
The following lemma states that the set of probability distributions of the empirical measures $\{\pi^N\}_{N\in\N}$ is tight on the Skorohod space. 
\begin{lemma}\label{lem:Tightness_of_PN}
	Let $T>0$ be fixed and $\mathbb{P}^N$ be given by Definition~\ref{def:P_N_and_E_N}.
Then the set $\{\mathbb{P}^N\}_{N\in\N}$ is pre-compact on $D([0,T]\to \mathcal{M}_+)$.
Moreover, any limiting distribution $\mathbb{P}^\infty$ is supported on $C([0,T],\M)$.

\end{lemma}
\begin{proof}
To show tightness of $\{\mathbb{P}^N\}_{N\in\N}$ it is sufficient to check tightness of $\{\mathbb{P}^N\circ G^{-1}\}_{N\in\N}$ for a dense subset of \eqref{eq:smooth_functions} (cf.~\cite{KL1999,Mit1983}). 
Thus, it suffices to prove that the set of probability distributions of $\{\<\pi^N(t),G\>\}_{t\in[0,T]}$ is  pre-compact for any $G$ which belongs to 
\begin{equation}\label{eq:smooth_functions_dense_subset}
	\mathcal{D}_0 := \{ G\in \mathcal{D} \,\big\vert\, G\in C^\infty([-1,2]),\ \supp G\subset (-1,0)\cup(0,1)\cup(1,2) \}.
\end{equation}
By the Dynkin martingale formula, there  exists a martingale $M_t^N(G)$ such that 
\begin{equation}\label{eq:Dynkin_martingale_formula}
	\<\pi^N(t),G\> = \<\pi^N(0),G\> + \int[0][t] N^2 L_N \<\pi^N(s),G\> \dif s + M_t^N(G). 
\end{equation}
By a version of the Prohorov theorem by Aldous (see \cite[Theorem~1.3,\,Proposition~1.6]{KL1999}), and in view of \eqref{eq:Dynkin_martingale_formula}, it remains to prove the following estimates for all $t>0$ (cf.~\cite[Theorem~2.1]{KL1999}),
\begin{align}
	\sup_{N\in\N} &\mathbb{E}^N \big[ \sup_{0\leq s\leq t} (N^2 L_N \<\pi^N(s),G\>)^2 \big]<\infty, \label{eq:estimate1} \\
	\lim_{N\to\infty} &\mathbb{E}^N \big[ \sup_{0\leq s\leq t} M^N_s(G)^2 \big] =0. \label{eq:estimate2}
\end{align}

Let us show first that \eqref{eq:estimate2} holds true.
The increasing process $\<M^N\>_t$ given by the Doob-Meyer decomposition satisfies
\begin{gather*}
	\pd{\phantom{.}}{t} \<M^N(G)\>_t = \Gamma^N(\eta^N(t),G), \\
\Gamma^N(\eta,G) = \frac{1}{N^2} \sum_{x=-N}^{2N-1} c_{x,x+1}(\eta) \big[\partial_N G(\frac{x}{N}) (\eta_{x+1}-\eta_x)\big]^2,
\end{gather*}
where $\partial_N G(\frac{x}{N}) := N[G(\frac{x+1}{N}) - G(\frac{x}{N})]$ is the discrete gradient on $\{\frac{x}{N}\}_{x\in\Z}$ and $\Gamma^N$ is the so-called ''carr\'{e} du champs``.
Since, for some $C>0$,
\[
	|\Gamma^N(\eta,G)| \leq  \frac{C}{N}, \quad N\in\N,
\]
then, by the Doob maximal inequality \eqref{eq:estimate2} follows.

Let us deal now with \eqref{eq:estimate1}.
We denote the discrete Laplacian on $ \{\frac{x}{N}\}_{x\in\Z}$ by 
\[
	(\partial^2_N G)(\frac{x}{N}) := N^2 (G(\frac{x+1}{N}) + G(\frac{x-1}{N}) - 2 G(\frac{x}{N})).
\]
The following discrete analogue of the Newton formula holds true, 
\[
\sum_{x=a}^b (\partial_N f)(x) (\partial_N g)(x) = - \sum_{x=a+1}^b f(x) (\partial^2_N g)(x) +f(b+1) (\partial_N g)(b) - f(a) (\partial_N g)(a).
\]
Then, by \eqref{assum:gradient_condition} and the Newton formula, for all $G\in\mathcal{D}_0$, as $N\to\infty$,

\begin{align}
	N^2 L_N \<\pi,G\>(\eta) &= - \sum_{x=-N}^{2N-1} c_{x,x+1}(\eta) (\eta_{x+1}-\eta_x) (\partial_N G)(\frac{x}{N}) \notag\\ 
											&= - \frac{1}{N} \big[ (\sum_{x=-N}^{-1} + \sum_{x=N+1}^{2N-1})  (\partial_N \Pi_x)(\eta) + \sum_{x=0}^{N} (\partial_N \tau_xh)(\eta) \big]  (\partial_N G)(\frac{x}{N}) \notag\\ 
											&= \frac{1}{N} \big[ (\sum_{x=-N}^{-1} + \sum_{x=N+1}^{2N-1})  \Pi_x(\eta) + \sum_{x=0}^{N} \tau_x h(\eta) \big]  (\partial^2_N G)(\frac{x}{N}). \label{eq:generator_regroup_G_0}
\end{align}
Hence \eqref{eq:estimate1} is bounded by $ C\!\int[-1][2] \partial_u^2 G(u)^2 \dif u$ and tightness of $\{\mathbb{P}^N\}_{N\in\N}$ is proven.

Proceeding similarly with
\[
	N^2 L_N \<\pi^N,G\>^2(\eta)-2\<\pi^N,G\>(\eta) L_N \<\pi^N,G\>(\eta),
\]
we conclude that its expectation with respect to $\mathbb{E}^N$ is bounded by $\frac{C}{N^2} (\int[-1][2] G'(u)^2\dif u)^2$ (see \cite[p.~127]{ELS1991}).
Hence, $\<\pi^N,G\>$ changes in a jump at most by $\frac{1}{N^2}\|G'\|_\infty$.
Therefore, the limiting distribution $\mathbb{P}^\infty$ is supported on the set of continuous trajectories $C([0,T],\M)$.
The proof is fulfilled.
\end{proof}

\subsection{Macroscopic equation}

Lemma~\ref{lem:Tightness_of_PN} says that there exists a limiting distribution $\mathbb{P}^\infty$ of $\{\mathbb{P}^N\}_{N\in\N}$. The following lemma ensures that $\mathbb{P}^\infty$ is supported on the absolutely continuous paths, which are weak solutions to \eqref{eq:PDE}.
\begin{lemma}\label{lem:Macroscopic_equation}
	For any limiting point $\pi^\infty$ of $\{\tilde{\pi}^{N,k}\}_{N\in\N}$ defined by \eqref{eq:average_empirical_density} there exists 
	\begin{equation}\label{eq:rho_in_L2_H1}
		\rho\in L^2([0,T],H^1(K)), \quad K\Subset (-1,0)\cup(0,1)\cup(1,2),\ T>0, 
	\end{equation} 
	such that the limiting distribution $\mathbb{P}^\infty$ is supported on the absolutely continuous paths:
	\begin{equation}\label{eq:limiting_density}
	\mathbb{P}^\infty \big[ \pi^\infty(\dif u,t ) = \rho(u,t)\dif u \,\big\vert\, \rho \text{ is a weak solution to } \eqref{eq:PDE}  \big] = 1,
\end{equation}
where $\rho_0$ is defined by \eqref{assum:lim_of_init_distrib} and a weak solution to \eqref{eq:PDE} is considered in the sense of Definition~\ref{def:very_weak_solution}.
\end{lemma}

\begin{proof}
Derivation of the macroscopic equation \eqref{eq:PDE} is standard.
See for instance \cite[pp.~128-129]{ELS1991}.
For convenience of the reader, we give a sketch of the proof pointing out important differences, thus making our exposition self-contained.

Similarly to \eqref{eq:generator_regroup_G_0}, for the truncated empirical density $\pi^{N,k}$ defined by \eqref{eq:truncated_empirical_density} we get for all $G\in\mathcal{D}$, 
\begin{align}
	N^2 L_N &\<\pi^{N,k},G\>(\eta) = \frac{1}{N} \big[ \sum_{x=k}^{N-k} \tau_x h(\eta)  + \sum_{x=-N+k}^{-k} \Pi_x + \sum_{x=N+k}^{2N-k}\Pi_x \big] G''(\frac{x}{N}) \notag \\   
														& +G'(-1) [k\Pi_{-N+k-1} - (k-1)\Pi_{-N+k}] + G'(0) [(k-1)\tau_{-k} h-k\tau_{-k+1} h] \notag \\
														& +G'(0) [k\tau_{k-1} h-(k-1)\tau_{k} h] + G'(1) [(k-1)\tau_{N-k} h-k\tau_{N-k+1} h] \notag \\
														& +G'(1) [k\Pi_{N+k-1} - (k-1)\Pi_{N+k}] + G'(2) [(k-1)\tau_{2N-k} h-k\tau_{2N-k+1} h] \notag \\
														& + O(\frac{1}{N}), \ N\to\infty. \label{eq:generator_regroup_G_k}
\end{align}
Note that this time $-1,\,0,\,1,\,2$ may belong to the support of $G$.
For this reason, the boundary terms with $G'(-1)$, $G'(0)$, $G'(1)$, $G'(2)$ appear. 

As before, we can show that the quadratic variation of the martingale given by the Dynkin martingale formula for $\<\pi^{N,k},G\>$ vanishes uniformly in $k\in\N\cup\{0\}$. 
Hence,
\begin{equation}\label{eq:expectation_vanishes}
	\lim_{N\to\infty} \mathbb{E}^N \big[ \sup_{0\leq t \leq T} \big\vert \<\pi^{N,k}(t),G\> -  \<\pi^{N,k}(0),G\> - \int[0][t] N^2 L_N \<\pi^{N,k}(s),G\> \dif s  \big\vert  \big] = 0.
\end{equation}

Now we want to utilize the local equilibrium property of the process $\{\eta^N(t)\}_{t\geq0}$.
It states that, in the limit, local spatial averages of a local function $g$ at $\eta^N(t)$ may be replaced by an expectation of $g$ with respect to the Gibbs measure at a local density of $\eta^N(t)$.

\begin{lemma}[Replacement Lemma]\label{lem:replacement}
	Let $\La_N$ equal to one of the intervals $[-N,-1]$, $[0,N]$, $[N+1,2N]$ on $\Z$, and $g:\Z\to\R$ be local. Then
	\[
		\lim_{\eps\to 0_+} \limsup_{N\to\infty} \sup_{\substack{\{x\in\Z:\, B_{\eps N}(x) \subset \La_N\}}} \int[0][t]\mathbb{E}^N  \big[ V_{\eps N}(x, \eta^N(s))\big] \dif s =0,
	\]
	where, 
	\begin{equation}\label{eq:V_eps_N}
		V_{\eps N}(x,\eta) := \Big\vert \frac{1}{|B_{\eps N}|} \sum_{y\in B_{\eps N}(x)} g(\tau_y\eta) - \<g\> (M_{B_{\eps N}(x)}(\eta)) \Big\vert.
	\end{equation}
\end{lemma}
In order to apply Lemma~\ref{lem:replacement} to \eqref{eq:expectation_vanishes} and \eqref{eq:generator_regroup_G_k}, we need an additional spatial averaging in $k$. For this reason, we consider $\tilde{\pi}^{N,k}(t)$ define dy \eqref{eq:average_empirical_density}.
Repeating for $\<\tilde{\pi}^{N,k}(t), G\>$ the arguments made for $\<\pi^{N,k}(t), G\>$, we obtain an analogue of \eqref{eq:expectation_vanishes} (see \cite[(3.30)]{ELS1991}).
As a result, by Lemma~\ref{lem:replacement}, we conclude \eqref{eq:limiting_density}. 

By \cite[Lemma~4]{ELS1990}, \cite[(3.34)]{ELS1991}, for any limiting measure $\mathbb{P}^\infty$ of $\{\mathbb{P}^N\}_{N\in\N}$, and all $T>0$, $K\Subset(-1,0)\cup(0,1)\cup(1,2)$, there exists a constant $C>0$, such that  
			\[
				\mathbb{E}^\infty \big[ \int[0][T] \int[K] |\partial_u \rho(u,t)|^2 \dif u \dif t \big] \leq CT. 
			\]
Moreover, since $\rho$ is a limit of empirical densities, we have $0\leq \rho \leq 1$. Thus, \eqref{eq:rho_in_L2_H1} follows and the proof is fulfilled.
\end{proof}

Note that in comparison to \cite{ELS1991} we do not now values of $\rho(u,s)$ on the boundaries. 
As we will see in the next subsections, in our case different methods are needed to derive boundary conditions at the points $u\in\{-1,0,1,2\}$.

\subsection{Continuity of the flow}
\begin{lemma}\label{lem:Neumann_boundary}
	Any limiting density $\rho(u,t)$ of the process $\{\eta^N(t)\}_{N\in\N}$ given by Lemma~\ref{lem:Macroscopic_equation} satisfies \eqref{eq:Neum} and \eqref{eq:Neumann_boundary} in the sense of Definition~\ref{def:weak_Neumann}.
\end{lemma}
\begin{proof}
For any $w\in(0,1)$, $r\in\R$ - small,
\begin{multline}
	\frac{1}{rN} \int[0][t] \mathbb{E}^N \big[ \frac{1}{2\eps N+1} \sum_{z=-\eps N}^{\eps N} \sum_{j=[wN]}^{[wN]+[rN]-1} \tau_{z+j} \partial_N h(\eta^N(s)) \big] \dif s\\ 
	= \frac{1}{r}\int[0][t] \mathbb{E}^N \big[ \frac{1}{2\eps N+1} \sum_{z=[wN]-\eps N}^{[wN]+\eps N} \big( \tau_{z+[rN]} h(\eta^N(s)) - \tau_{z} h(\eta^N(s)) \big)  \big]. \label{eq:macro_derivative}
\end{multline}
By Lemma~\ref{lem:Macroscopic_equation} and Lemma~\ref{lem:replacement}, as $N\to\infty$, $\eps\to 0_+$,  \eqref{eq:macro_derivative} converges to 
\begin{equation}\label{eq:dif_Phi}
	\frac{1}{r} \int[0][t] \big(\Phi(\rho(w+r,s)) - \Phi(\rho(w,s))\big) \dif s.
\end{equation}
Similarly, for $v\in(-1,0)$ and $\Pi(\eta)=\eta_0$ in place of $h(\eta)$, we obtain in the limit
\begin{equation}\label{eq:dif_Pi}
	\frac{1}{r} \int[0][t] \big(\rho(v+r,s) - \rho(v,s)\big) \dif s.
\end{equation}
By the Dynkin formula applied to $\Pi_x(\eta) = \eta_{x}$,
\[
	 	\mathbb{E}^N \big[ \eta_x^N(t) - \eta_x^N(0) \big] = N^2 \int[0][t] \mathbb{E}^N \big[ \mathcal{C}_x(\eta^N(s)) \big] \dif s, 
\]
where
\[
	\mathcal{C}_x(\eta) = c_{x,x+1}(\eta)(\eta_{x+1} - \eta_{x})-c_{x-1,x}(\eta)(\eta_{x} - \eta_{x-1}).
\]
Then, subtracting from the first formula in \eqref{eq:macro_derivative} the analogous one at $v\in(-1,0)$, we get
\begin{align*}
	\frac{1}{rN} \int[0][t] \mathbb{E}^N &\big[\frac{1}{2\eps N{+}1} \!\!\sum_{x=-\eps N}^{\eps N}\!\! \sum_{j=0}^{[rN]-1} \big( \tau_{x+j+[wN]} \partial_N h(\eta^N(s) - \tau_{x+j+[vN]} \partial_N \Pi(\eta^N(s) \big)  \big] \dif s \\
	&\quad = \frac{1}{r} \int[0][t] \mathbb{E}^N \big[ \frac{1}{2\eps N{+}1} \sum_{x=-\eps N}^{\eps N} \sum_{j=0}^{[rN]-1} \sum_{l=[vN]}^{[wN]-1} \mathcal{C}_{x+j+l}(\eta^N(s)) \big] \dif s\\
	&\quad = \frac{1}{rN^2} \mathbb{E}^N \Big[ \frac{1}{2 \eps N{+}1} \sum_{x=-\eps N}^{\eps N} \sum_{j=0}^{[rN]-1} \sum_{l=[vN]}^{[wN]} \big[ \eta^N_{x+j+l}(t) - \eta^N_{x+j+l}(0) \big]  \Big]. 
\end{align*}
Since $\eta^N_{x+j+l}\in\{0,1\}$, the absolute value of the latter formula is bounded by $w-v+\frac{1}{N}$.
Hence, the absolute value of the difference of \eqref{eq:dif_Phi} and \eqref{eq:dif_Pi} is bounded by $w-v$. 
By Definition~\ref{def:weak_Neumann}, taking $w\to0_+$, $v\to0_-$, we prove the first equality in \eqref{eq:Neumann_boundary}.
The proof of the second equality is analogous and we omit it.
\end{proof}

\subsection{Continuity of the chemical potential}
\begin{lemma}\label{lem:Dirichlet_boundary}
	Any limiting density $\rho(u,t)$ of the process $\{\eta^N(t)\}_{N\in\N}$ given by Lemma~\ref{lem:Macroscopic_equation} satisfies \eqref{eq:Dirichlet_boundary} for almost all $t>0$.
\end{lemma}
\begin{proof}
	The statement of the lemma follows from the two block estimate at the macroscopic points $u=0$ and $u=1$. Namely, \eqref{eq:two_block_estimate_for_potentials} and \eqref{eq:two_block_estimate_for_potentials_ii} in Lemma~\ref{lem:two_block_estimate} below together with Lemma~\ref{lem:Macroscopic_equation} imply for all bounded continuous functions $F:\R\to\R_+$, 
\begin{equation*}
	\int[0][t] |F(\la^+_{\rho(s,0_+)})-F(\la^-_{\rho(s,0_-)})|\dif s =0, \quad  \int[0][t] |F(\la^-_{\rho(s,1_+)})-F(\la^+_{\rho(s,1_-)})|\dif s =0, \quad t>0,
\end{equation*}
from which \eqref{eq:Dirichlet_boundary} immediately follows.
\end{proof}

\section{The Law of Large Numbers for Gibbs Measures}

\subsection{Existence and uniqueness of Gibbs measures}\label{sec:exist_unique_Gibbs}
The following proposition shows that for any $\la\in\R$, each of the sequences $\{\nu_K^{\la,\pm}\}_{K\Subset\Z}$ defined by \eqref{eq:Gibbs_grandcanonical_homogeneous} as well as its counterparts with a boundary condition as in \eqref{eq:Gibbs_grandcanonical_homogeneous_boundary} have a unique limiting point, as $K$ tends to $\Z$.
\begin{proposition}\label{prop:Gibbs_measure_exist_unique}
	For any $\la\in\R$, there exist unique probability measures $\nu^{\la,+}$, $\nu^{\la,-}$ on $(\iX,\F)$, called Gibbs measures, such that for all local functions $g:\Z\to\R$ and all boundary conditions $\omega\in\iX$, the following limits hold true
	\[
		\lim_{K\uparrow\Z}\int[\iX_K] g(\eta) \nu^{\la,\pm}_K(\eta\,\vert\,\omega)\dif\eta =  \lim_{K\uparrow\Z}\int[\iX_K] g(\eta) \nu^{\la,\pm}_K(\eta)\dif\eta = \int[\iX] g(\eta) \nu^{\la,\pm}(\dif\eta).
	\]
	Moreover, $\{\nu^{\la,\pm}(\cdot\,\vert\,\omega)\}_{K\Subset\Z}$ are specifications for $\nu^{\la,\pm}$, namely, the following consistency holds, 
	\begin{equation}\label{eq:consistency}
		\nu^{\la,\pm}(\eta\,\vert\,\omega) := \nu^{\la,\pm}(\eta\,\vert\,\F_{K^c})(\omega) = \nu^{\la,\pm}_K (\eta \,\vert\, \omega), \quad \eta\in\iX_K,\ \omega\in\iX.
	\end{equation}
\end{proposition}
\begin{proof}
	We will prove the statement for $\nu^{\la,+}$.
In notations of \cite{Geo2011}, the state space of a configuration at a fixed point $\frac{x}{N}$ is described by the finite set $E=\{0,1\}$.
	Then, by \cite[Example~4.11,\,(2)]{Geo2011}, for any $\omega\in\iX$, the sequence $\{\nu_K^{\la,+}(\cdot\,\vert\,\omega)\}_{K\Subset\Z}$ is compact.
Thus, existence of a cluster point  $\nu^{\la,+}$ follows.
Since the Hamiltonian in \eqref{eq:Hamiltonian_with_boundary_condition} is defined by a finite range potential  
	\begin{equation}\label{eq:potential}
		\Psi_A(\eta) =
		\left\{
		\begin{aligned}
			&q \eta_x \eta_{x+1}, & &A=\{x,x+1\},\\
			&\la \eta_x, & &A=\{x\},\\
			&0, & &\text{otherwise},
		\end{aligned}
		\right.
	\end{equation}
	then, by \cite[Thereom~4.17]{Geo2011}, the cluster points satisfy \eqref{eq:consistency}.
	Next, by \cite[Theorem~8.39]{Geo2011}, for any fixed $\la\in\R$, the cluster point $\nu^{\la,+}$ is unique and independent of $\omega\in\iX$.
	Finally, \cite[Example~4.20,\,(1)]{Geo2011} implies that the Gibbs distributions $\{\nu^{\la,+}_K\}_{K\Subset\Z}$ with free boundary condition converge to $\nu^{\la,+}$ as well.
	For $\nu^{\la,-}$ the proof is the same and we omit it.
\end{proof}

We will often need to compare a Gibbs distribution with free boundary condition and a one with a boundary condition $\omega$.
To do this we remove interaction with the boundary and estimate the error.
For example, if $K$ is an interval $[a,b]\subset\Z$, then we have two bonds connecting the boundary:  $[a-1,a]$ and $[b,b+1]$.
Each of the bonds contributes to the Gibbs distribution \eqref{eq:Gibbs_grandcanonical_homogeneous} at most $\Err=e^{2|q|}$ (cf. \eqref{eq:Hamiltonian_with_free_condition} and~\eqref{eq:Hamiltonian_with_boundary_condition}).
Hence the following estimate holds true   
\begin{equation}\label{eq:boundary_vs_free_boundary}
	\Err^{-2} \nu^{\la,+}_{[a,b]}(\eta) \leq \nu^{\la,+}_{[a,b]}(\eta\,\vert\,\omega) \leq \Err^2 \nu^{\la,+}_{[a,b]}(\eta), \quad \eta\in \iX_{[a,b]}.
\end{equation}
Similarly, for any $c\in(a,b)$, $\eta\in\iX_{[a,c]}$, $\tilde{\eta}\in\iX_{[c,b]}$,
\begin{equation}\label{eq:split_interval}
	\Err^{-2} \nu^{\la,+}_{[a,c]}(\eta)\, \nu^{\la,+}_{[c,b]}(\tilde{\eta}) \leq \nu^{\la,+}_{[a,b]}(\eta\*\tilde{\eta}) \leq \Err^2 \nu^{\la,+}_{[a,c]}(\eta)\, \nu^{\la,+}_{[c,b]}(\tilde{\eta}).
\end{equation}

\subsection{Equivalence of Ensembles}\label{sec:equivalence_of_ensembles}
	
Additionally to the grand canonical partition functions $Z^{\la,\pm}_K$ with free boundary defined in \eqref{eq:Gibbs_grandcanonical_homogeneous} let us consider the canonical partition functions (with free boundary), 
	\begin{equation}\label{eq:partition_function_canonical}
		Z_{K,n}^\pm := \sum_{\eta\in\iX_K} \1_{\{N_K(\eta)=n\}}\, e^{-H_K^\pm(\eta)},
	\end{equation}
where $H^-_K\equiv0$ and $H^+_K$ equals $H_K$ in \eqref{eq:Hamiltonian_with_free_condition}.
In the following proposition we remind the reader the equivalence of ensembles at the level of potentials.
\begin{proposition}[{\cite[Proposition~3.9]{Geo1979}}]\label{prop:equivalence_of_ensembles}
	Suppose that, $\la\in\R$, $\rho\in [0,1]$, $\rho_l := \frac{n_l}{2l} \to \rho$, $l\to\infty$. Then the following limits exist
	\begin{equation}\label{eq:def_q_p}
		q^\pm(\rho) := \lim_{l\to\infty} \frac{\log Z^\pm_{B_l,n_l}}{2l}, \quad p^\pm(\la) := \lim_{l\to\infty} \frac{\log Z_{B_l}^{\la,\pm}}{2l}.
	\end{equation}
The functions $q^\pm$ are concave and $p^\pm$ are convex.
The following relations between $p^\pm$ and $q^\pm$ hold true
	\begin{align}
		q^\pm(\rho) &= \inf_{\la\in\R} [\, p^\pm(\la) - \rho\la \, ], \qquad \rho\in[0,1], \label{eq:q_inf_p}\\
		p^\pm(\la)  &= \max_{\rho\in[0,1]} [\, q^\pm(\rho) + \rho\la \,], \qquad \la\in\R. \label{eq:p_max_q}
	\end{align}
	Moreover, for any $\rho \in[0,1]$ there exist unique $\la_\rho^\pm\in[-\infty,\infty]$ such that
	\begin{equation}\label{eq:equation_rho_la}
		p^\pm(\la_\rho^\pm) - \rho \la_\rho^\pm - q^\pm(\rho) = 0.
	\end{equation}
\end{proposition}

For any local function $g$ on $\{0,1\}^\Z$ and $\rho\in[0,1]$, we denote 
\begin{equation}\label{eq:expectation_wrt_canonical_Gibbs}
	\<g\>^\pm(\rho) := \<\nu^{\la_\rho^\pm,\pm},g\> = \int[\{0,1\}^\Z] g(\eta) \nu^{\la_\rho^\pm,\pm}(\dif \eta),
\end{equation}
where $\la_\rho^\pm$ are defined by Proposition~\ref{prop:equivalence_of_ensembles}.
For $h$ from \eqref{assum:gradient_condition} we define  
	\begin{equation}\label{eq:expectation_homogeneous}
		\Phi(\rho):= \<h\>^{+}(\rho). 
\end{equation} 

Note that $\<\Pi\>^-(\rho)= \rho$, for $\Pi(\eta)=\eta_0$. That is why a linear equation is obtained on $(-1,0)\cup(1,2)$ in \eqref{eq:PDE}.

\subsection{The Law of Large Numbers}
We denote for $g:\Z\to\R$ - local, 
\begin{equation}\label{eq:V_l}
	V_l^\pm(x,\eta) := \big\vert \frac{1}{2l+1} \sum_{y=-l}^{l} \tau_y g(\eta) - \<g\>^\pm(M_{B_l(x)}(\eta)) \big\vert.
\end{equation}
First, let us remind the reader the following ergodic theorem for translation invariant Gibbs measures.
\begin{lemma}[{\cite[(4.16)]{FHU1991}}]\label{lem:LLN_Gibbs_translation_invariant}
	For any $g:\Z\to\R$ - local,
	\[
		\lim_{l\to\infty} \sup_{\rho\in[0,1]} \int[\iX] V_l^\pm(0,\eta) \nu^{\la^\pm_\rho,\pm}(\dif\eta) =0. 
	\]
\end{lemma}

The following lemma states that the law of large numbers is valid for Gibbs measures.
\begin{lemma}\label{lem:LLN_weak}
	For any $\#\in\{-,+\}$, $\delta>0$, there exist $\sigma>0$, $L\in\N$, such that,
	\[
		\sup_{\rho\in[0,1]} \nu^{\la_\rho^\#,\#}(|M_{B_l}-\rho|\geq \delta) \leq e^{-\sigma l}, \quad l\geq L.
	\]
\end{lemma}

\begin{proof}
	Let us first consider the densities separated from $\rho=0$ and $\rho=1$. We fix $\gamma\in(0,\frac{1}{2})$.
	The following limit is continuous in $\rho\in(0,1)$, and as a result uniformly continuous on $[\gamma,1-\gamma]$,
\begin{equation}\label{eq:pressure_uniform_in_rho}
	p^\#(\la_\rho) := \lim_{l\to\infty} \frac{\log Z^{\la^\#_\rho,\#}_{B_l}}{|B_l|}.
\end{equation}
Since $\rho\to\la^\#_\rho$ and $\la\to Z^{\la,\#}_{|B_l|}$ are increasing and continuous, then $\rho\to \log Z^{\la^\#_{\rho},\#}_{B_l}$ is increasing and continuous as well.
Therefore, by Dini's theorem, the convergence in \eqref{eq:pressure_uniform_in_rho} holds locally uniformly on $(0,1)$.
Hence, for any $\eps>0$, there exist $\delta_1\in(0,\delta)$, $L\in\N$, such that 
\begin{equation}\label{eq:some_uniform_estimate}
	\frac{Z^{\la_{\rho_2}^\#,\#}_{B_l}}{Z^{\la_{\rho_1}^\#,\#}_{B_l}} \leq e^{|B_l| (p^\#(\la_{\rho_2}^\#) -  p^\#(\la_{\rho_1}^\#)+\eps)} \leq e^{2\eps|B_l|} , \quad l\geq L,\ \rho_1,\rho_2\in[\gamma,1-\gamma],\ |\rho_2-\rho_1|<\delta_1.
\end{equation}
	Next, by \eqref{eq:consistency} and \eqref{eq:boundary_vs_free_boundary}, for $\rho_1,\rho_2\in[\gamma,1-\gamma]$, $\rho_1\in[\rho_2-\delta_1,\rho_2]$,
	\begin{align}
		\nu^{\la_{\rho_1}^\#,\#}(|M_{B_l}-\rho_1|\geq\delta) &= \int \nu^{\la_{\rho_1}^\#,\#}_{B_l}(|M_{B_l}-\rho_1|\geq\delta \,\vert\, \omega) \nu^{\la_{\rho_1}^\#,\#}(\dif\omega) \nonumber \\
		 &\leq \Err^{2} \nu^{\la_{\rho_1}^\#,\#}_{B_l}(|M_{B_l}-\rho_1|\geq\delta)\leq \Err^{2} \nu^{\la_{\rho_1}^\#,\#}_{B_l}(|M_{B_l}-\rho_2|\geq\delta-\delta_1) \nonumber\\ 
		 &\leq \Err^{2} \frac{Z^{\la_{\rho_2}^\#,\#}_{B_l}}{Z^{\la_{\rho_1}^\#,\#}_{B_l}} \nu^{\la_{\rho_2}^\#,\#}_{B_l}(|M_{B_l}-\rho_2|\geq\delta-\delta_1), \label{eq:LDP_estimate}
	\end{align}
where in the last inequality we used that $\la_{\rho_1}^\# \leq \la_{\rho_2}^\#$ for $\rho_1\leq\rho_2$.

	To show the large deviation principle for $\nu^{\la_{\rho_2}^\#,\#}$, we apply \cite[Theorem~II.6.1]{Ell2007} to the random variables $W_{|B_l|} = |B_l|\rho_2 - N_{B_l}$ on the probability spaces
\[
	\Omega_{|B_l|} = \{0,1\}^{B_l},\ \F=2^{\Omega_{|B_l|}},\ P_{|B_l|}(\cdot) = \nu^{\la_{\rho_2}^\#,\#}(\cdot).
\]
Then the upper large deviation bound is valid with the rate function equal to the convex conjugate of the following limit
\begin{multline}
	\frac{1}{|B_l|} \log \<e^{\mu W_{|B_l|}}, P_{|B_l|}\> = \mu \rho_2 + \frac{\log Z^{\la_{\rho_2}^\#-\mu,\#}_{B_l}}{|B_l|} - \frac{\log Z^{\la_{\rho_2}^\#,\#}_{B_l}}{|B_l|}\\ \to \mu\rho_2 + p^\#(\la_{\rho_2}-\mu) - p^\#(\la_{\rho_2}), \quad l\to\infty.
\end{multline}
It is easy to check that the rate function is separated from zero on $\R\backslash(\delta_1-\delta,\delta-\delta_1)$. Therefore, there exist $\sigma>0$, $L\in\N$, such that 
\[
	\nu^{\la_{\rho_2}^\#,\#}_{B_l}(|M_{B_l}-\rho_2|\geq\delta-\delta_1) \leq e^{-\sigma |B_l|}, \quad l\geq L.
\]
Since the cover of $[\gamma,1-\gamma]$ by $[\rho-\delta_1,\rho]$ with $\rho\in [\gamma+\delta_1,1-\gamma]$ has a finite sub-cover, then taking $2\eps<\sigma$ and redefining $\sigma$ and $L$ we conclude the statement of the lemma for $\rho\in[\gamma,1-\gamma]$.

Let $\gamma<\delta$. Then, for any $\rho\in[0,\gamma]$, by Chebyshev's inequality, following \eqref{eq:LDP_estimate},
\begin{align}\label{eq:some_estimate}
	\nu^{\la_{\rho}^\#,\#}(|M_{B_l}-\rho|\geq\delta) &= \nu^{\la_{\rho}^\#,\#}(N_{B_l}\geq(\delta+\rho)|B_l|) \leq \Err^{2} \< e^{ N_{B_l}-(\rho+\delta)|B_l|},  \nu^{\la_{\rho}^\#,\#}_{B_l} \>.
\end{align}
By \eqref{eq:Hamiltonian}, 
\[
	- |q| N_{B_l}(\eta) \leq q \sum_{x,x+1\in B_l} \eta_x \eta_{x+1} \leq |q|N_{B_l}(\eta).
\]
Thus, we may continue in \eqref{eq:some_estimate},
\begin{align}
	\leq \Err^2 e^{-(\rho+\delta)|B_l|} \frac{ \sum_{\eta\in\iX_{B_l}} e^{(\la_\rho+|q|)N_{B_l}} }{ \sum_{\eta\in\iX_{B_l}} e^{(\la_\rho-|q|)N_{B_l}} }  \leq  \Err^2 e^{-|B_l|(\rho+\delta) + |B_l| \log(1+e^{\la_\gamma+|q|}) }.
\end{align}
We remind that $\la_\gamma\downarrow-\infty$, as $\gamma\downarrow 0$. Therefore, we may choose $\gamma$ such that 
\[
	\delta - \log(1+e^{\la_\gamma+|q|}) >0,
\]
which implies the statement of the lemma for $\rho\in[0,\gamma]$. Due to the symmetry with respect to the change of variables $\eta_x\to 1-\eta_x$, we conclude the statement for $\rho\in[1-\gamma,1]$. The proof is fulfilled.

\end{proof}

\section{Entropy Estimates}
In the present section we introduce various Dirichlet forms and establish its properties needed for the proof of the replacement lemma (Lemma~\ref{lem:replacement}) and the proof of the boundary conditions \eqref{eq:Dirichlet_boundary} (see proof of Lemma~\ref{lem:Dirichlet_boundary}).

Let $\mu_t^N(\dif\eta)$ denote the probability distribution of the process $\{\eta^N(t)\}_{t\geq0}$ defined by the generator $N^2 L_N$ in \eqref{eq:generator}.
Throughout the section we will need only the following two facts: there exists a Hamiltonian $H$ such that the corresponding Gibbs measure $\nu^\la_N$ defined by \eqref{eq:Gibbs} and \eqref{eq:abuse_notations_of_Gibbs} is an invariant measure of the process, and the Hamiltonian $H$ satisfies the detailed balance condition \eqref{assum:DB}.
In particular, proofs in this section do not rely on the explicit definitions \eqref{eq:nonlocal_exchange_rates}, \eqref{eq:simple_symmetric_exchange_rates}, \eqref{eq:exchange_rates_on_the_boundaries}.

For $x,x+1\in [-N,2N]$, we introduce a quadratic form 
\begin{equation}
	D_{x,x+1}^N[\mu^N_t] := \frac{1}{2} \sum_{\eta\in\iX_N}  \Big[\sqrt{ c_{x,x+1}(\eta^{x,x+1})\mu^N_t(\eta^{x,x+1})} - \sqrt{ c_{x,x+1}(\eta)\mu^N_t(\eta)} \,\Big]^2 \label{eq:def_D_x_x+1^N},
\end{equation}
where we abused notations: $\mu^N_t(\eta):=\mu^N_t(\{\eta\})$.
For $K\subset [-N,2N]$, we put
\[
	D^N_K[\mu^N_t] := \sum_{x,x+1\in K} D^N_{x,x+1}[\mu^N_t]. 
\]
Then, the entropy production of $\mu^N_t$ equals 
\[
	D^N[\mu^N_t] := D^N_{[-N,2N]}[\mu^N_t] = \sum_{x,x+1\in [-N,2N]} D^N_{x,x+1}[\mu^N_t].
\]
To distinguish exchange rates of the particles inside $[0,N]$, and ones inside $[-N,-1]$ and $[N+1,2N]$, we write $c_{x,x+1}^+$ for $c_{x,x+1}$ defined by \eqref{eq:nonlocal_exchange_rates}, and $c_{x,x+1}^-$ for $c_{x,x+1}$ defined by \eqref{eq:simple_symmetric_exchange_rates}.
The, for $\mu\in\M(\iX)$ and $x,x+1\in K \subset\Z$, we introduce the quadratic form 
\begin{equation}
	D_{x,x+1}^{K,\#}[\mu] := \frac{1}{2} \sum_{\eta\in\iX_{K_r}}  \Big[\sqrt{ c_{x,x+1}^\#(\eta^{x,x+1})\mu\sVert[1]_{\F_{K_r}}\!\!\!(\eta^{x,x+1})} - \sqrt{ c_{x,x+1}^\#(\eta)\mu\sVert[1]_{\F_{K_r}}\!\!\!(\eta)} \,\Big]^2 \label{eq:D^K_f^K},
\end{equation}
where $\#\in\{+,-\}$, $r$ is defined by item~\ref{item:generality} of Remark~\ref{rem:main_remark}, and $K_r$ by \eqref{eq:K_r}.   
The entropy production of $\mu$ on $K$ is defined as follows
\[
	D^{K,\#}[\mu] = \sum_{x,x+1\in K} D^{K,\#}_{x,x+1}[\mu],\quad \#\in\{-,+\}.
\]
The following generator corresponds to the Dirichlet form $D^{K,\#}$ (see Lemma~\ref{lem:variational_representation_of_entropy_production} below),  
\[
	(L_{K}^{\#} g) (\ka,\xi) := \sum_{x,x+1\in K} c_{x,x+1}^{\#}(\eta)(g(\eta^{x,x+1})-g(\eta)).
\]
We extend $\bar{\mu}^N_t\in\M(\iX_N)$ to $\M(\iX)$ by \eqref{eq:mu_N_extended}.
Then, the entropy productions $D^{K,\#}$ of $\mu^N_t$ on $K\subset\Z$ is well-defined.
In particular, 
\begin{align*}
	D^N_{x,x+1}[\mu_t^N] &= D^{K,+}_{x,x+1}[\mu_t^N],\quad x,x+1\in K,\ K_r\subset[0,N],\\  
	D^N_{x,x+1}[\mu_t^N] &= D^{K,-}_{x,x+1}[\mu_t^N],\quad x,x+1\in K,\ K_r\subset[-N,-1]\cup[N+1,2N].
\end{align*}
In order to prove the two block estimate (see Lemma~\ref{lem:two_block_estimate} below),  we will need an analogue of $D^{K,\#}$ for $\mu\in\M(\iX^2)$. For $\#_1,\#_2\in\{-,+\}$,
\begin{align}
	&D^{K,\#_1,\#_2}[\mu] := \frac{1}{2} \sum_{\ka,\xi\in\iX_{K_r}} \sum_{x,x+1\in K} \notag\\  
								 &\Big[\Big(\sqrt{ c_{x,x+1}^{\#_1}(\ka^{x,x+1})\mu\sVert[1]_{\F_{K_r}\times\F_{K_r}}\!\!\!(\ka^{x,x+1},\xi)} - \sqrt{ c_{x,x+1}^{\#_1}(\ka)\mu\sVert[1]_{\F_{K_r}\times\F_{K_r}}\!\!\!(\ka,\xi)} \,\Big)^2 \notag\\
								 &+\Big(\sqrt{ c_{x,x+1}^{\#_2}(\xi^{x,x+1})\mu\sVert[1]_{\F_{K_r}\times\F_{K_r}}\!\!\!(\ka,\xi^{x,x+1})} - \sqrt{ c_{x,x+1}^{\#_2}(\ka)\mu\sVert[1]_{\F_{K_r}\times\F_{K_r}}\!\!\!(\ka,\xi)} \,\Big)^2 \Big], \label{eq:def_D^KK}
\end{align}
with the corresponding generator 
\begin{multline}
	(L_{K}^{\#_1,\#_2})g (\ka,\xi) :=\sum_{x,x+1\in K} c_{x,x+1}^{\#_1}(\ka)(g(\ka^{x,x+1},\xi)-g(\ka,\xi)) \\ + c_{x,x+1}^{\#_2}(\xi)(g(\ka,\xi^{x,x+1})-g(\ka,\xi)).
\end{multline}
In order to compare densities of the particles in two blocks $B_l(x_0)$ and $B_l(z_0)$, for $|x_0-z_0|\leq 2\eps N$, $x_0+l<z_0-l$, we introduce the Dirichlet form
\begin{multline}\label{eq:entropy_production_x_z}
	\D_{x,z}^{K,\#_1,\#_2}[\mu] := \frac{1}{2} \sum_{\ka,\xi\in\iX_{K_r}}  \Big[\sqrt{ \c_{x,z}^{\#_1,\#_2}(\ka^x,\xi^z)\mu\sVert[1]_{\F_{K_r}\times\F_{K_r}}\!\!\!(\ka^x,\xi^z)} \\- \sqrt{ \c_{x,z}^{\#_1,\#_2}(\ka,\xi)\mu\sVert[1]_{\F_{K_r}\times\F_{K_r}}\!\!\!(\ka,\xi)} \,\Big]^2, 
\end{multline}
where
\begin{equation}\label{eq:c_xz}
	\c_{x,z}^{\#_1,\#_2}(\ka,\xi) := [\ka_x(1-\xi_z) + \xi_z(1-\ka_x)] e^{-\frac{1}{2}(\Delta_x H^{\#_1}_{K_r}(\ka)+ \Delta_z H^{\#_2}_{K_r}(\xi))}.
\end{equation}
Here $H^{-}_{K_r}\equiv0$ and $H^{+}_{K_r}$ is defined by \eqref{eq:Hamiltonian_with_free_condition}.
Note that $\c_{x,z}^{\#_1,\#_2}$ satisfies an analogue of the detailed balance condition 
\begin{equation}\label{eq:DB_x_z_analogue}
	\c_{x,z}^{\#_1,\#_2}(\ka^x,\xi^z) = \c_{x,z}^{\#_1,\#_2}(\ka,\xi) e^{\Delta_x H_{K_r}^{\#_1}(\ka) + \Delta_z H_{K_r}^{\#_2}(\xi)}.
\end{equation}
The corresponding generator is defined as follows,
\[
	(\L_{x,z}^{\#_1,\#_2} g)(\ka,\xi) := \c_{x,z}^{\#_1,\#_2}(\ka,\xi) (g(\ka^x,\xi^z) - g(\ka,\xi)).
\]

\begin{lemma}\label{lem:variational_representation_of_entropy_production}
	For $x,x{+}1{\in}[-N,2N]$, the following variational formulas hold true
	\begin{align}
	D_{x,x+1}^N [\mu^N_t] &= \sup\, \{ - \int[\iX_N] \frac{L_{x,x+1} u}{u} \dif\mu^N_t \,\big\vert\, u>0\ is\ \F_{N} - measurable\} \label{eq:D_x_x+1^N}\\
	D_K^N [\mu^N_t] &= \sup\, \{ - \int[\iX_N] \frac{\sum_{x,x+1\in K}L_{x,x+1} u}{u} \dif\mu^N_t \,\big\vert\, u>0\ is\ \F_{N} - measurable\} \label{eq:D_K^N}\\
	D^N [\mu^N_t] &= \sup\, \{ - \int[\iX_N] \frac{L_N u}{u} \dif\mu^N_t \,\big\vert\, u>0\ is\ \F_{N} - measurable\} \label{eq:D^N}
	\end{align}
	For $\mu{\in}\M(\iX)$, $x,x{+}1{\in}K$, the following variational formulas hold true
	\begin{align}
	D_{x,x+1}^{K,\pm} [\mu] &= \sup\, \{ - \int[\iX] \frac{L_{x,x+1}^\pm u}{u} \dif\mu \,\big\vert\, u>0\ is\ \F_{K_r} - measurable\} \label{eq:D_x_x+1^K}\\
	D^{K,\pm} [\mu] &= \sup\, \{ - \int[\iX] \frac{L_K^\pm u}{u} \dif\mu \,\big\vert\, u>0\ is\ \F_{K_r} - measurable\} \label{eq:D^K}
	\end{align}
	For $\mu\in\M(\iX^2)$, $x,z\in K$, $\#_1,\#_2\in\{-,+\}$, the following variational formulas hold true
	\begin{align}
		D^{K,\#_1,\#_2} [\mu] &= \sup\, \{ - \int[\iX^2] \frac{L_{K}^{\#_1,\#_2} g}{g} \dif\mu \,\big\vert\, g>0\ is\ \F_{K_r}\times\F_{K_r} - measurable\} \label{eq:D^KK}\\
		\D_{x,z}^{K,\#_1,\#_2} [\mu] &= \sup\, \{ - \int[\iX^2] \frac{\L_{x,z}^{\#_1,\#_2} g}{g} \dif\mu \,\big\vert\, g>0\ is\ \F_{K_r}\times\F_{K_r} - measurable\} \label{eq:D_x_z^KK}
	\end{align}
\end{lemma}
\begin{proof}
	Let us prove \eqref{eq:D_x_x+1^N} and \eqref{eq:D_K^N}.
First, we note that for any $\mu\in\M(\iX_N)$, splitting the following integral into two and applying to one of them the change of variables $\eta\to\eta^{x,x+1}$ we get the estimate
\begin{multline*}
	- \int[\iX_N] c_{x,x+1}(\eta) \big(\frac{u(\eta^{x,x+1})}{u(\eta)}-1\big) \mu(\dif\eta) = \frac{1}{2} \sum_{\eta\in\iX_N} \big( c_{x,x+1}(\eta) \mu(\eta) + c_{x,x+1}(\eta^{x,x+1})\mu(\eta^{x,x+1}) \\
	- c_{x,x+1}(\eta)\frac{u(\eta^{x,x+1})}{u(\eta)} \mu(\eta)  - c_{x,x+1}(\eta^{x,x+1})\frac{u(\eta)}{u(\eta^{x,x+1})} \mu(\eta^{x,x+1})  \big)  \leq D_{x,x+1}^N(\mu),
\end{multline*}
which implies the inequality ``$\geq$'' in \eqref{eq:D_x_x+1^N}.
Taking sum over $x,x+1\in K$ in the previous estimate, we get ``$\geq$'' in \eqref{eq:D_K^N}.
Next, taking $u(\eta) = \sqrt{\mu(\eta) e^{H(\eta)}}$ and applying \eqref{assum:DB} we conclude the equality in \eqref{eq:D_x_x+1^N} and \eqref{eq:D_K^N}.
The proof of other variational formulas is similar with the following exception: To prove ``$\leq$'' in \eqref{eq:D_x_z^KK} we consider $g(\ka,\xi) = \sqrt{\mu(\ka,\xi)e^{H(\ka)+H(\xi)}}$ and apply \eqref{eq:DB_x_z_analogue} instead of \eqref{assum:DB}. We leave details to the reader.
\end{proof}

The relative entropy of $\mu_t^N$ with respect to the Gibbs measure $\nu^\la_N$ is defined as follows,
\[
	H(\mu_t^N\,\vert\,\nu^\la_N) = \int[\iX_N] f_t(\eta) \ln f_t(\eta) \nu^\la_N(\dif\eta), \qquad f_t(\eta) := \diff{{\mu^N_t}}{{\nu^\la_N}}(\eta), 
\]
where $f_t(\eta)$ implicitly depends on $N$.  Let us denote,
\begin{equation}\label{eq:mu_N}
	\bar{\mu}^N_t := \int[0][t] \mu^N_s \dif s.
\end{equation}
The following lemma states in particular, that at any moment of time the probability distributions of the exclusion process $\{\eta^N(t)\}_{t\geq0}$ are not too far from the Gibbs measures $\nu^\la_N(\eta)\dif\eta$.
\begin{lemma}\label{lem:entropy_estimates}
	There exists a constant $C>0$, such that for all $N\in\N$, $t\geq0$,
\[
	H(\mu_t^N\,\vert\, \nu^\la_\La) \leq CN,\quad \text{and}\quad D(\bar{\mu}_t^N) \leq {\frac{C}{N}}.
\]
\end{lemma}

\begin{proof}
First, we prove the estimate of the relative entropy at $t=0$. By \eqref{eq:Gibbs}, 
\[
	\frac{1}{Z^\la_{N}}e^{-3(|q|+|\la|)N}	\leq \nu^\la_{N} (\eta) \leq \frac{1}{Z^\la_{N}}e^{3(|q|+|\la|)N}, \quad \eta\in\iX_N.
\]
Similarly to \eqref{eq:def_q_p}, by the sub-additivity argument there exists a finite limit of $\frac{\ln{ Z^\la_N}}{N}$ as $N$ tends to infinity.
Therefore, we obtain for some $C>0$ (cf.~\cite[Lemma~2.10]{Fun2018})
\[
	H(\mu_0^N\,\vert\, \nu^\la_N) =\!\! \int[\iX_N] (\log\mu_0^N(\eta) - \log\nu^\la_N(\eta))  \mu_0^N(\dif\eta) \leq \int[\iX_N] (\ln Z_N^\la + 3(|q|+|\la|)N) \mu_0^N(\dif\eta)  \leq CN.
\]

Next, we estimate the relative entropy for $t>0$,
\begin{align*}
	\diffp{}t H(\mu_t^N\,\vert\,\nu^\la_N) &= \int[\iX] \diffp{f_t}{t}(\eta) \ln f_t(\eta) \nu^\la_N(\dif\eta)\\ 
																				 &= N^2 \sum_{x=-N}^{2N-1} \int[\iX_N] \big[ f_t(\eta) (L_{x,x+1} \ln f_t)(\eta) - (L_{x,x+1}f_t)(\eta)\big]\nu^\la_N(\dif\eta).
\end{align*}
Since, $a(\ln b-\ln a) - (b-a) \leq - (\sqrt{b}-\sqrt{a})^2$, then
\[
	f_t(\eta) (L_{x,x+1}\ln f_t)(\eta) - (L_{x,x+1}f_t)(\eta) \leq -c_{x,x+1}(\eta) \big[ \sqrt{f_t(\eta^{x,x+1})} - \sqrt{f_t(\eta)}\, \big]^2.
\]
By \eqref{assum:DB},
\[
	D^N_{x,x+1}(\mu_t^N) = \frac{1}{2}\int[\iX_N] c_{x,x+1}(\eta) \big[ \sqrt{f_t(\eta^{x,x+1})} - \sqrt{f_t(\eta)}\,\big]^2 \nu^\la_N(\dif\eta).
\]
Therefore, we conclude
\[
\frac{\partial}{\partial t} H(\mu_t^N\,\vert\,\nu^\la_N) \leq -N^2 D^N(\mu_t^N). 
\]
Hence, by convexity of $D^N$ as a function of $\mu^N_t$, integrating over $[0,t]$ we obtain,
\[
	H(\mu_t^N\,\vert\, \nu^\la_\La) + t N^2  D^N(\bar{\mu}_t^N)\leq H(\mu_0^N\,\vert\,\nu^\la_\La) \leq CN.
\]
The proof is fulfilled.
\end{proof}

\begin{lemma}\label{lem:entropy_production_estimate}
	There exist $N_0\in\N$ and $C>0$, such that for all $t>0$, $N\geq N_0$,
	\begin{align*}
		D^{K,+}[\bar{\mu}_t^N] &\leq \frac{C}{N}, \qquad K_r\subset [0,N], \\ 
		D^{K,-}[\bar{\mu}_t^N] &\leq \frac{C}{N}, \qquad K_r\subset [-N,0]\cup[N+1,2N],
	\end{align*} 
\end{lemma}
\begin{proof}
Let $K_r\subset[0,N]$. Since the class of $\F_{K_r}$-measurable functions is a subset of $\F_N$-measurable function, then by Lemma~\ref{lem:variational_representation_of_entropy_production},
	\begin{align*}
		D^{K,+}[\bar{\mu}^N_t] &= \sup\, \{ - \int[\iX_N] \frac{L_K^+ u}{u} \dif\bar{\mu}^N_t \,\big\vert\, u>0\ is\ \F_{K_r} - measurable\} \\
										 &\leq \sup\, \{ - \int[\iX_N] \frac{L_K^+u}{u} \dif\bar{\mu}_t^N \,\big\vert\, u>0\ is\ \F_N - measurable\} \\
			&\leq \sum_{x\,and\,x+1\in K} \sup\, \{ - \int[\iX_N] \frac{L_{x,x+1} u}{u} \dif\bar{\mu}^N_t \,\big\vert\, u>0\ is\ \F_N - measurable\} \\
			&= \sum_{x\,and\,x+1\in K} D_{x,x+1}^N[\bar{\mu}^N_t] \leq D^N [\bar{\mu}^N_t].
	\end{align*}	
	The same inequality holds true if $K_r\subset[-N,-1]\cup[N+1,2N]$.
	By Lemma~\ref{lem:entropy_estimates}, we conclude the statement of the lemma.
\end{proof}

\section{Proof of the replacement lemma (Lemma~\ref{lem:replacement})}

We follow the standard approach and split the proof into one- and two-block estimates.

In the case of the translation invariant space (the torus $\mathbb{T}_N$ instead of the interval $\La_N$) a common way to prove the one-block estimate (see e.g. \cite{FHU1991}) is to show an ergodic representation for a limiting measure $\mu^\infty$ of $\{\bar{\mu}^N_t\}$ in terms of the extreme invariant measures of the process, which, in our case, are the Gibbs measures.
To justify application of the ergodic representation theorem, one needs to prove existence of a probability measure $w$ on the class of Gibbs measures, which specifies which convex combination of the Gibbs measures equals to the limiting measure $\mu^\infty$: 
\[
	\mu^\infty(\,\cdot\,) = \int[0][1] \nu^{\la_\rho}(\,\cdot\,) w(\dif\rho).
\]
To prove existence of $w$ one may apply the Birkhoff ergodic theorem, which requires additional averaging of $\bar{\mu}^N_t$ over spatial translations: $\tilde{\mu}^N_t := \frac{1}{|\mathbb{T}_N|}\sum_{x\in\mathbb{T}_N} \tau_x\circ\bar{\mu}^N_t$.
Then, $\tau_x \circ \tilde{\mu}^N_t = \tilde{\mu}^N_t$ allows to conclude the statement.

Our main difficulty is that we consider the interval $\La_N$ instead of the torus and the Birkhoff theorem can not be applied. 
Therefore, we fail to prove the ergodic representation before the one-block estimate is shown.
On the other hand we may modify approach in \cite[Subsection~5.4]{KL1999}, which does not require to prove existence of the measure $w$. This will allow us to conclude the one block estimate (see Lemma~\ref{lem:one_block_estimate}).

After the one-block estimate is proven, we can use it instead of the Birkhoff theorem to prove existence of $w$ and conclude the ergodic representation theorem (see Lemma~\ref{lem:ergodic_representation}).  This will allows us to prove the two-block estimate similarly to the standard approach in \cite{GPV1988}. 

The form of the one block estimate, which we need, is also presented in \cite{ELS1990}. However, some details of the proof remained unclear to us. In particular, rigorous verifications of the law of large numbers for Gibbs measures \textit{uniformly} in $\rho\in[0,1]$ (see Lemma~\ref{lem:LLN_weak}) seems to be absent in the literature, as well as the treatment of the interactions on the boundary of a block (cf.~\eqref{eq:boundary_vs_free_boundary} and \eqref{eq:split_interval}).


\begin{lemma}[One block estimate]\label{lem:one_block_estimate}
	Let $\La_N$ equal to one of the intervals $[-N,-1]$, $[0,N]$, $[N{+}1,2N]$ on $\Z$. Then, for any local function $g:\Z\to\R$, $\supp g\subset B_r$, and $V^\pm_l$ defined by \eqref{eq:V_l},
	\begin{equation}\label{eq:one_block_estimate_lim}
		\lim_{l\to\infty} \limsup_{N\to\infty} \sup_{\substack{\{x\in\Z:\, B_{l+r}(x) \subset \La_N\}}} \int[\iX_N] V^\pm_{l}(x, \eta)\, \bar{\mu}^N_t(\dif\eta) =0. 
	\end{equation}
\end{lemma}

\begin{proof}
	Let $\La_N=[0,N]$. The proof of other two cases is similar and we omit it.
	We remind the reader that $\mu^N_t$ is extended on $\iX$ by \eqref{eq:mu_N_extended}. Then, the integral in \eqref{eq:one_block_estimate_lim} equals to
	\[
		\int[\iX] V_l^+(0,\eta) (\tau_{-x}\circ\bar{\mu}^N_t) (\dif\eta) =  \int[\iX_{B_{l+r}}] V_l^+(0,\eta) (\tau_{-x}\circ\bar{\mu}^N_t)\sVert[2]_{\F_{B_{l+r}}}\!\!\!\!\!\!\! (\dif\eta),
	\]
	which by Lemma~\ref{lem:entropy_production_estimate} is bounded from above by
	\begin{equation}\label{eq:sup}
		\sup\{\int[\iX_{B_{l+r}}] V_l^+(0,\eta) \mu(\dif\eta) \,\big\vert\, \mu\in\M(\iX_{B_{l+r}}):\,D^{B_{l+r},+}(\mu)\leq \frac{C}{N}\}.
	\end{equation}
	Since $\M(\iX_{B_{l+r}})$ is compact in the weak topology and $D^{B_{l+r},+}$ is lower-semicontinuous, then for any $N\in\N$, the supremum \eqref{eq:sup} is reached by some $\mu^N_{B_{l+r}}\in \M(\iX_{B_{l+r}})$, and there exists a limiting probability measure $\mu^\infty_{B_{l+r}}$ for $\{\mu^N_{B_{l+r}}\}_{N\in\N}$.
	Moreover, by lower semicontinuity of the entropy production $D^{B_{l+r},+}$, we obtain 
	\[
		D^{B_{l+r},+}(\mu^\infty_{B_{l+r}})=0.
	\]
Therefore, by \eqref{assum:DB} and \eqref{eq:D^K_f^K}, for $\eta\in\iX_{B_l}$, $\omega\in\iX_{B_{l+r}\backslash B_l}$, $x,x+1\in B_l$,
	\begin{equation}\label{eq:mu_satisfies_DB}
		\mu^\infty_{B_{l+r}}((\eta\*\omega)^{x,x+1}) = \mu^\infty_{B_{l+r}}(\eta\*\omega) e^{-\Delta_{x,x+1} H^+_{B_{l+r}}(\eta\cdot\omega)}.
	\end{equation}
	Any state $\tilde{\eta}\in\iX_{B_l}$ may be obtained from $\eta$ by a superposition of the maps $\eta \to \eta^{x,x+1}$, $x,x+1\in B_l$, provided $N_{B_l}(\eta) = N_{B_l}(\tilde{\eta})$.
	Therefore, \eqref{eq:mu_satisfies_DB} uniquely defines a probability measure in $\M(\iX_{B_{l+r}})$, given $\omega\in\iX_{B_{l+r}\backslash {B_l}}$ and $N_{B_l}(\eta)$.
On the other hand, the Gibbs distribution $\nu^{\la,+}_{B_{l+r}}$ satisfies \eqref{eq:mu_satisfies_DB} as well.
Hence, for $j=0,1,\cdots,2l+1$, 
	\begin{equation}\label{eq:mu_infty_is_locally_Gibbs}
		\mu^\infty_{B_{l+r}}(\eta\*\omega\,\vert\,\omega, N_{B_l}(\eta)=j) = \nu^{\la,+}_{B_{l+r}}(\eta\*\omega\,\vert\,\omega, N_{B_l}(\eta)=j). 
	\end{equation}
	Note, that the right hand side in \eqref{eq:mu_infty_is_locally_Gibbs} does not depend on $\la$.
	Let $f^\infty$ be the density of $\mu^\infty_{B_{l+r}}$ with respect to $\nu^{\la,+}_{B_{l+r}}$, namely,
	\[
		\mu^\infty_{B_{l+r}}(\dif\eta) = f^\infty(\eta) \nu^{\la,+}_{B_{l+r}}(\dif\eta), \quad \eta\in\iX_{B_{l+r}}.
	\]
	By \eqref{eq:mu_infty_is_locally_Gibbs}, for any $\omega\in\iX_{B_{l+r}\backslash B_l}$, $f^\infty$ is constant on the hyperplanes 
\[
	\{\eta\*\omega \,\vert\, \eta\in\iX_{B_l},\,N_{B_l}(\eta)=j\},\quad j=0,1,...,2l+1.
\]
To simplify notations we put (c.f.~\eqref{eq:abuse_notations_of_Gibbs}),
\begin{align*}
	\nu := \nu^{\la,+}_{B_{l+r}}, \qquad \nu^{j}(\,\cdot\,) := \nu^{\la,+}_{B_{l+r}}(\,\cdot\, \vert \, \omega,\,N_{B_l}(\eta)=j).
\end{align*}
Note that $M_{B_l}\in[0,1]$ if and only if $N_{B_l}\in[0,2l+1]$.
By the tower property of conditional expectation, 
\begin{align}
\int[\iX_{B_{l+r}}] V^+_{l} &\dif\mu^\infty_{B_{l+r}} \leq \sum_{j=0}^{2l+1}\, \int[\iX_{B_{l+r}}] \1_{\{N_{B_l}=j\}}\, V_l^+ f^\infty \dif\nu \notag\\ 
																										 &= \sum_{j=0}^{2l+1}\, \int[\iX_{B_{l+r}}] \int[\iX_{B_{l+r}}] \1_{\{N_{B_l}=j\}}\, V_l^+ f^\infty \dif\nu^{j} \dif\nu \notag\\ 
																										 &= \sum_{j=0}^{2l+1}\, \int[\iX_{B_{l+r}}] \1_{\{N_{B_l}=j\}}\, f^\infty \dif\nu \int[\iX_{B_{l+r}}] V_l^+ \dif\nu^{j} \leq \sup_{j\in [0,2l+1]} \int[\iX_{B_{l+r}}] V_l^+ \dif\nu^j, \label{eq:tower_property}
\end{align}
where we applied 
\[
	\sum_{j \in [0,2l+1]}\, \int[\iX_{B_{l+r}}] \1_{\{N_{B_l}=j\}}\, f^\infty \dif\nu \leq 1.
\]
Thus, to prove the lemma, it is sufficient to show, that the letter integral in \eqref{eq:tower_property} vanishes uniformly in $j$, as $l\to\infty$.

For $k\in\N$, $r<k<l$, consider a disjoint cover of $\Z$ by the balls $B_k^i := B_k(i+2ki)$, $i\in\Z$.
Let $I\subset \Z$ be the maximal set such that $B_k^i\subset B_{l-r}$, $i\in I$.
We write $\tilde{B}$ for $B_l\backslash \bigcup\limits_{i\in I} B_k^i$.
As a result $B_l$ is a disjoint union of $\tilde{B}$ and  $B_k^i$, $i\in I$.
Then,
\[
	\int[\iX_{B_{l+r}}] V^+_l(0,\eta)\, \nu^j (\dif\eta) \leq \frac{|B_k|}{|B_l|} \sum_{i\in I} \int[\iX_{B_{l+r}}] V_k^i  \dif\nu^j + O(\frac{k}{l}), \quad l\to\infty, 
\]
where 
\[
	V_k^i := \big\vert \frac{1}{|B_k|} \sum_{y\in B_k^i} \tau_y g - \<g\>^+(M_{B_k^i}) \big\vert + \big\vert \<g\>^+(M_{B_k^i}) - \<g\>^+(M_{B_l}) \big\vert,
\]
and we used that
\[
	\frac{|\tilde{B}|}{|B_l|} \sum_{i\in I} \int[\iX_{B_{l+r}}] \big\vert \frac{1}{|\tilde{B}|} \sum_{y\in \tilde{B}} \tau_y g - \<g\>^+(M_{B_l}) \big\vert \dif\nu^j \leq \frac{4k}{l} \|g\|_\infty. 
\]

Now we want to compare $\int V_k^i \dif \nu^j$ for different $i\in I$.
Denote $S_i := \supp V_k^i$, $S_k^c := B_l \backslash S_i$.
Without loss of generality we may assume that $S_i$ is an interval on $\Z$.
By \eqref{eq:split_interval}, we may remove interactions between $B_{l+r}\backslash B_l$, $S_i$ and $S_i^c$, compensating it by errors $|\Err_j|\leq e^{4|q|}$, $j=1,2$:
\begin{multline*}
	\int[\iX_{B_{l+r}}] V_k^i\, \1_{\{N_{B_l}=j,\,\Pi_{B_{l+r}\backslash B_l}=\omega\}} \dif\nu_{B_{l+r}}^{\la,+} = \Err_1 \nu^{\la,+}_{B_{l+r}\backslash B_l}(\omega) \int[\iX_{B_l}] \1_{\{N_{B_l}=j\}} V_k^i \dif\nu_{B_l}^{\la,+} \\
																																																								= \Err_1 \Err_2 \nu^{\la,+}_{B_{l+r}\backslash B_l}(\omega) \sum_{l=0}^j \int[\iX_{S_i}] \1_{\{N_{S_i}=l\}} V_k^i\, \dif\nu^{\la,+}_{S_i} \int[\iX_{S_i^c}] \1_{\{N_{S_i^c}=j-l\}} \dif\nu^{\la,+}_{S_i^c}.
\end{multline*}
Next, in the last integral we glue together the disjoint components of $S_i^c$ with an error $|\Err_3|\leq e^{4|q|}$, and shift the obtained interval:
\[
	\int[\iX_{S_i^c}] \1_{\{N_{S_i^c}=j-l\}} \dif\nu_{S_i^c}^{\la,+} = \Err_3 \int[\iX_{[1,|S_i^c|]}] \1_{\{N_{[1,|S_i^c|]}=j-l\}} \dif\nu_{[1,|S_i^c|]}^{\la,+}. 
\]

By translation invariance, $\int \1_{\{N_{S_i}=l\}} V^i_k \dif\nu^{\la,+}_{S_i}$ and $|S_i^c|$ do not depend on $i\in I$. Therefore, for any fixed $i_*\in I$,
\[
	\int V_k^i \dif \nu^j = \Err_1 \Err_2 \Err_3 \int V_k^{i_*} \dif \nu^{j}, \quad i\in I.
\]

Hence, to prove the lemma it is sufficient to show 
\begin{multline}
	\lim_{k\to\infty} \lim_{l\to\infty} \sup_{j\in [0,2l+1]} \Big[ \int \big\vert \frac{1}{|B_k|} \sum_{y\in B_k^{i_*}} \tau_y g - \<g\>^+(M_{B_k^{i_*}}) \big\vert \dif\nu^j \\ + \int \big\vert \<g\>^+(M_{B_k^{i_*}}) - \<g\>^+(M_{B_l}) \big\vert \dif\nu^j \Big] =0. \label{eq:final_limit_in_one_block_est}
\end{multline}
The first integral in \eqref{eq:final_limit_in_one_block_est} vanishes by the equivalence of ensembles \cite[Corollary~7.13]{Geo1979} and Lemma~\ref{lem:LLN_Gibbs_translation_invariant}. The second integral in \eqref{eq:final_limit_in_one_block_est} vanishes by continuity of $\rho\to\<g\>(\rho)$ and Lemma~\ref{lem:LLN_weak}.
The proof is fulfilled.
\end{proof}
For $C>0$, $\eps\in(0,1)$, $\#_1,\#_2\in\{-,+\}$, we denote 
\begin{equation}\label{eq:A_eps}
\A_\eps^{\#_1,\#_2} := \{\mu\in\M(\iX^2)\,\vert\,	D^{K,\#_1,\#_2}[\mu] =0,\ \D_{x,z}^{K,\#_1,\#_2}[\mu] \leq C \eps,\ x,\,z\in K\Subset\Z \}.
\end{equation}
Let us also denote $\hat{\tau}_{x,z} \eta := (\tau_{-x}\eta,\tau_{-z}\eta)$. 
\begin{lemma}\label{lem:A_eps}

	Let $\#_1,\#_2\in\{-,+\}$, and $\La_N^{\#_1}$, $\La_N^{\#_2}$ equal to one of the intervals $[-N,-1]$, $[0,N]$, $[N+1,2N]$ (possibly different).
Then, for any $C>0$, as $N\to\infty$, all limiting measures of 
	\[
		\{\bar{\mu}^N_t\circ\hat{\tau}_{x_0,z_0}^{-1}\,\vert\, N\in\N,\ B_l(x_0)\subset\La_N^{\#_1},\ B_l(z_0)\subset\La_N^{\#_2},\ |x_0-z_0|\leq2\eps N\},
	\] 
	belong to $\A_\eps^{\#_1,\#_2}$.
\end{lemma}
\begin{proof}
	First note that for any $\F_{K_r}\times\F_{K_r}$\,-\,measurable function $g$, $u(\eta):= g(\tau_{-x_0}\eta,\tau_{-z_0}\eta)$ is $\F_\La$\,-\,measurable with $\La:=(K_r+x_0)\cup(K_r+z_0)$.
In particular, $u$ is $\F_N$\,-\,measurable for sufficiently large $N$.
 Then \eqref{eq:D_K^N}, \eqref{eq:D^KK} and Lemma~\ref{lem:entropy_estimates} imply 
	\[
	D^{K,\#_1,\#_2}[\bar{\mu}^N_t\circ\hat{\tau}_{x_0,z_0}^{-1}] \leq D_{\La}^N[\bar{\mu}^N_t]\leq D^{N}[\bar{\mu}^N_t] \leq \frac{C}{N}.
	\]
	Taking $N\to\infty$, we conclude the first condition in \eqref{eq:A_eps}.

	To prove the second condition, we apply the so called telescopic argument. Namely, repeating line by line the proof of \cite[Lemma~4.1]{FHU1991} with $\hat{\tau}_{x_0,z_0}$ in place of $\hat{\tau}_{0,z}$,
	and applying Lemma~\ref{lem:entropy_estimates} we conclude that there exist constants $\tilde{C},\,C>0$ such that 
	\begin{align*}
		\D_{x,z}^{K,\#_1,\#_2} [\bar{\mu}^N_t] 
																	 &\leq \tilde{C}|x-z| \sum_{y=x}^{z-1} D^N_{y,y+1}[\bar{\mu}_t^N] \leq \tilde{C} (2\eps N+4l) D^N[\bar{\mu}_t^N] \leq C\eps. 
	\end{align*}
	The proof is fulfilled.
\end{proof}

\begin{lemma}[Ergodic~representation]\label{lem:ergodic_representation}
	For any $\mu\in\A_\eps^{\#_1,\#_2}$, $\#_1,\#_2\in\{-,+\}$, there exists a probability measure $w_\mu$ on $[0,1]^2$, such that, for any local functions $g_1,\, g_2$, 
  \begin{equation}\label{eq:ergodic_representation}
		\int[\iX^2] g_1(\ka) g_2(\xi) \mu(\dif\ka \dif\xi) = \int[[0,1]^2] \<g_1\>^{\#_1}(\rho_1)\,\<g_2\>^{\#_2}(\rho_2) w_\mu(\dif\rho_1 \dif\rho_2).
  \end{equation}
\end{lemma}
\begin{proof}
	The proof follows from the one of \cite[Proposition~4.2]{FHU1991}, where instead of applying the Birkhoff theorem (which the authors call the individual ergodic theorem), we use the one-block estimate (Lemma~\ref{lem:one_block_estimate}). 
\end{proof}

\begin{lemma}[Two-block~estimate]\label{lem:two_block_estimate}
	For $\La_N\in\{[-N,-1]$, $[0,N]$, $[N+1,2N]\}$, and any $F:[0,1]\to\R_+$ - bounded continuous, as $N\to\infty$, $l\to\infty$, $\eps\to0_+$, 
  \begin{equation}
		\max_{\substack{\{x_0,\,z_0:\,|x_0-z_0|\leq 2\eps N\\ B_l(x_0)\cup B_l(z_0)\subset \La_N\}}} \int[\iX] |F(M_{B_l(x_0)}(\eta)) -  F(M_{B_l(z_0)}(\eta))| \bar{\mu}^N_t(\dif\eta) \to 0.\label{eq:two_block_estimate_for_densities}
  \end{equation}
	Moreover, for any $F:\R\to\R_+$ - bounded continuous, as $N\to\infty$, $l\to\infty$, $\eps\to0_+$, 
	\begin{align}
		\max_{\substack{\{x_0\in [-\eps N,-l)\\z_0\in [l,\eps N]\}}} \int[\iX] |F(\la^-_{M_{B_l(x_0)}(\eta)}) - F(\la^+_{M_{B_l(z_0)}(\eta)})| \bar{\mu}^N_t(\dif\eta) &\to 0, \label{eq:two_block_estimate_for_potentials} \\
		\max_{\substack{\{x_0\in [N+1-\eps N,N+1-l)\\z_0\in [N+1+l,N+1+\eps N]\}}} \int[\iX] |F(\la^+_{M_{B_l(x_0)}(\eta)}) -  F(\la^-_{M_{B_l(z_0)}(\eta)})| \bar{\mu}^N_t(\dif\eta) &\to 0. \label{eq:two_block_estimate_for_potentials_ii}
  \end{align}
\end{lemma}
The proof of the two-block estimate relies on several auxiliary lemmas.

\begin{lemma}\label{lem:lim_in_X2}
	For any $\#_1,\#_2\in\{-,+\}$, $F_1,F_2:[0,1]\to\R_+$ - bounded continuous,
	\begin{multline}\label{eq:lim_eq_sup}
		\lim_{l\to\infty} \sup_{\mu\in\A^{\#_1,\#_2}_{\eps}} \int[\iX^2] |F_1(M_{B_l}(\ka))-F_2(M_{B_l}(\xi))|^2 \mu(\dif\ka\dif\xi)\\
		= \sup_{\mu\in\A^{\#_1,\#_2}_{\eps}} \int[[0,1]^2] |F_1(\rho_1)-F_2(\rho_2)|^2 w_\mu(\dif\rho_1\dif\rho_2).
	\end{multline}
\end{lemma}
\begin{proof}
	By Lemma~\ref{lem:ergodic_representation}, the integral on the left hand side in \eqref{eq:lim_eq_sup} equals to (cf.~\eqref{eq:expectation_wrt_canonical_Gibbs})
	\begin{multline}\label{eq:int_eq_int}
		\int[\iX^2] \big[ \<F_1(M_{B_l}(\ka))^2\>^{\#_1}(\rho_1) -2 \<F_1(M_{B_l}(\ka))\>^{\#_1}(\rho_1) \<F_2(M_{B_l}(\xi))\>^{\#_2}(\rho_2)\\ 
		+\<F_2(M_{B_l}(\xi))^2\>^{\#_2}(\rho_2) \big] w_\mu(\dif\rho_1,\dif\rho_2).
	\end{multline}
	The first summand in \eqref{eq:int_eq_int} equals to 
	\begin{multline*}
	\int[[0,1]^2] \< F_1(M_{B_l}(\ka))^2\>^{\#_1}(\rho_1) \1_{\{|M_{B_l}(\ka)-\rho_1|\leq \delta\}} w_\mu(\dif\rho_1\dif\rho_2) \\
										+ O(\|F_1\|_\infty^2 \big[ \sup_{\rho_1\in[0,1]} \nu^{\la^{\#_1}_{\rho_1},\#_1}(|M_{B_l}(\ka)-\rho_1|>\delta) \big]), 
	\end{multline*}
	which converges uniformly in $\mu\in\A^{\#_1,\#_2}_{\eps}$ to $\int[[0,1]^2] F_1(\rho_1)^2\ w_\mu(\dif\rho_1\dif\rho_2)$ as $l\to\infty$, $\delta\to0$.
Repeating the argument for other summands in \eqref{eq:int_eq_int} we conclude \eqref{eq:lim_eq_sup}.
\end{proof}

\begin{lemma}\label{lem:supp_of_w_mu}
	For any limiting measure $\mu_0$ of $\A_{\eps}^{\#_1,\#_2}$, $\#_1,\#_2\in\{-,+\}$ (as $\eps\to0_+$), support of $w_{\mu_0}$, which is given by Lemma~\ref{lem:ergodic_representation}, satisfies the inclusion
	\begin{equation}\label{eq:supp_w_la1_la2}
		\supp w_{\mu_0} \subset \{(\rho_1,\rho_2)\in[0,1]^2 \,\vert\, \la^{\#_1}_{\rho_1}=\la^{\#_2}_{\rho_2}\}.
	\end{equation}
	In particular, if $\#_1=\#_2$, then
	\begin{equation}\label{eq:supp_w_rho1_rho2}
		\supp w_{\mu_0} \subset \{(\rho_1,\rho_2)\in[0,1]^2 \,\vert\, \rho_1 = \rho_2\}.
	\end{equation}
\end{lemma}
\begin{proof}
	Let $\#_1=-$, $\#_2=+$.
	By Lemma~\ref{lem:variational_representation_of_entropy_production}, for any $\mu\in\A_{\eps}^{-,+}$ and 
\[
	g(\ka,\xi) = e^{\frac{al}{2}(\alpha_1 M_{B_l}^2(\ka)+ \alpha_2 M_{B_l}^2(\xi))}, \quad a>0,\ \alpha_1,\alpha_2\in\R, 
\] 
the following inequality holds true uniformly in $x,\,z\in B_l$,
\[
	-\int[\iX^2] \c_{x,z}^{-,+}(\ka,\xi)\big(\frac{g(\ka^x,\xi^z)}{g(\ka,\xi)} -1\big) \mu(\dif\ka\dif\xi) = - \int[\iX^2] \frac{\L_{x,z}^{-,+} g}{g} \dif\mu \leq  \D_{x,z}^{B_l,-,+}[\mu]  \leq C \eps. 
\]
Since $|M_{B_l}|\leq1$, $\sup_{x,z} \|\c_{x,z}^{-,+}\|_\infty<\infty$, $e^{x}-1\leq e^{sign(x)}x$ for $|x|\leq1$, and 
\begin{multline*}
	\frac{al}{2}\big(\alpha_1 M_{B_l}^2(\ka^x)-\alpha_1 M_{B_l}^2(\ka) + \alpha_2 M_{B_l}^2(\xi^z)-\alpha_2 M_{B_l}^2(\xi)\big)\\ = -a(\xi_z-\ka_x)(\alpha_2 M_{B_l}(\xi)-\alpha_1 M_{B_l}(\ka)) + \frac{a(\alpha_1+\alpha_2)}{l}, \quad \ka_x\neq\xi_z,
\end{multline*}
then, 
\begin{multline}\label{eq:some_integral}
	\frac{1}{(2(l-r)+1)^2} \sum_{x,\,z\in B_{l-r}} \int[\iX^2] \c_{x,z}^{-,+}(\ka,\xi) (\xi_z-\ka_x)(\alpha_2 M_{B_l}(\xi)-\alpha_1 M_{B_l}(\ka)) \mu(\dif\ka\dif\xi) \\ \leq \frac{C \eps e}{a} + \frac{(|\alpha_1|+|\alpha_2|)e}{l}, \quad l\to\infty,\ \eps\to0_+,\ a\to0_+.
\end{multline}
To simplify notations, we denote $\la_1 := \la_{\rho_1}^-$, $\la_2 := \la_{\rho_2}^+$.
By Lemma~\ref{lem:ergodic_representation} and consistency of Gibbs distributions \eqref{eq:consistency}, the left hand side of \eqref{eq:some_integral} equals to (cf.~\eqref{eq:abuse_notations_of_Gibbs}), 
\begin{multline}\label{eq:ugly_integral}
	\frac{1}{(2l-2r+1)^2} \sum_{x,\,z\in B_{l-r}} \int[[0,1]^2] \int[\iX^2] \int[\iX_{B_{l}}] \int[\iX_{B_{l}}] \c_{x,z}^{-,+}(\ka,\xi) (\xi_z-\ka_x)(\alpha_2 M_{B_l}(\xi)-\alpha_1 M_{B_l}(\varkappa)) \times \\
	\times \nu_{B_{l}}^{\la_1,-}(\dif\varkappa\,\vert\,\omega^1) \nu_{B_{l}}^{\la_2,+}(\dif\xi\,\vert\,\omega^2) \nu^{\la_1,-}(\dif\omega^1) \nu^{\la_2,+}(\dif\omega^2) w_\mu(\dif\rho_1,\dif\rho_2).
\end{multline}
Hence, we need to compute,
\begin{multline*}
	\int[\iX_{B_{l}}] \int[\iX_{B_{l}}] \big(\1_{\{\varkappa_x=1,\,\xi_z=0\}} + \1_{\{\varkappa_x=0,\,\xi_z=1\}} \big) \c_{x,z}^{-,+}(\ka,\xi) (\xi_z-\ka_x)\times \\ 
	  \times (\alpha_2 M_{B_l}(\xi)- \alpha_1 M_{B_l}(\varkappa)) \nu_{B_{l}}^{\la_1,-}(\dif\varkappa\,\vert\,\omega^1) \nu_{B_{l}}^{\la_2,+}(\dif\xi\,\vert\,\omega^2),
	\end{multline*}
	which by definition of $\c_{x,z}^{-,+}$, equals to 
	\begin{multline*}
	(e^{\la_2}-e^{\la_1}) \int[\iX_{B_{l}}] \int[\iX_{B_{l}}] \1_{\{\varkappa_x=0,\,\xi_z=0\}} (\alpha_2 M_{B_l}(\xi)- \alpha_1 M_{B_l}(\varkappa)) \times \\ 
\times e^{-\frac{1}{2}\Delta_z H_{B_{l}}^+(\xi\cdot\omega^2)}\nu_{B_{l}}^{\la_1,-}(\dif\varkappa\,\vert\,\omega^1) \nu_{B_{l}}^{\la_2,+}(\dif\xi\,\vert\,\omega^2) + O(\frac{1}{l}), \quad l\to\infty.
\end{multline*}
Since $\Delta_z H_{B_{l}}^+$ is a local function, so are $g(\ka) := 1-\ka_0$, $\tilde{g}(\xi) := (1-\xi_0) e^{-\frac{1}{2} \Delta_0 H^+_{B_{l}}}$, and \eqref{eq:ugly_integral} equals to 
\begin{multline}
	\frac{e^{\la_2}-e^{\la_1}}{(2l-2r+1)^2} \sum_{x,\,z\in B_{l-r}} \int[[0,1]^2] \int[\iX^2] (\alpha_2 M_{B_l}(\xi)-\alpha_1 M_{B_l}(\varkappa)) 	\tau_x g(\varkappa\*\omega^1) \tau_z\tilde{g}(\xi\*\omega^2) \times \\
		\times \nu^{\la_1,-}(\dif\omega^1) \nu^{\la_2,+}(\dif\omega^2) w_\mu(\dif\rho_1,\dif\rho_2)+O(\frac{1}{l}),\quad l\to\infty.
\end{multline}
Finally, similarly to the proof of Lemma~\ref{lem:lim_in_X2}, as $l\to\infty$, $\eps\to0_+$, $a\to0_+$, we get for the limiting measure $\mu_0$,
\begin{equation}\label{eq:leq0}
	\int[[0,1]^2]	\<g\>^{-}(\rho_1) \<\tilde{g}\>^{+}(\rho_2)(\alpha_2\rho_2-\alpha_1\rho_1)(e^{\la_2}-e^{\la_1})\, w_{\mu_0}(\dif\rho_1 \dif\rho_2) \leq 0.
\end{equation}
Since $\<g\>^-(\rho_1) \<\tilde{g}\>^+(\rho_2)>0$, $\alpha_2,\alpha_1\in\R$ are arbitrary, and $\rho\to\la^{\pm}_\rho$ are strictly monotone, we conclude that \eqref{eq:supp_w_la1_la2} holds true.

The proof of the case $\#_1=+$, $\#_2=-$ is the same. 

Let $\#_1=\#_2$.
Since $\rho\to\la^{\#_1}_\rho$ is strictly monotone, then \eqref{eq:supp_w_rho1_rho2} follows from \eqref{eq:supp_w_la1_la2}. 
The proof is fulfilled.
\end{proof}

Now we can prove the two-block estimate.
\begin{proof}[{Proof of Lemma~\ref{lem:two_block_estimate}}]
	The limit \eqref{eq:two_block_estimate_for_densities} follows from the H\"{o}lder inequality applied to the integral in \eqref{eq:two_block_estimate_for_densities}, Lemma~\ref{lem:lim_in_X2} with $F_1=F_2=F$, and Lemma~\ref{lem:supp_of_w_mu}.
The same for \eqref{eq:two_block_estimate_for_potentials} with $F_1(\rho)=F(\la^-_\rho)$ and $F_2(\rho)=F(\la^+_\rho)$.
For \eqref{eq:two_block_estimate_for_potentials_ii}, consider $F_1(\rho)=F(\la^+_\rho)$ and $F_2(\rho)=F(\la^-_\rho)$.
\end{proof}

Finally, we are ready to prove the replacement lemma.
\begin{proof}[Proof of Lemma~\ref{lem:replacement}.]
	Let $J\subset\Z$ be the minimal set, for which $\{B_l^j:= B_l(j+2lj)\,\vert\,j\in J\}$ covers $B_{\eps N}(x)$.
	For simplicity we assume that $B_{\eps N}(x) = \sqcup_{j\in J} B_l^j$.
	Then, by \eqref{eq:V_eps_N}, 
	\begin{align}\label{eq:estimate_V_eps_N}
		V_{\eps N}(x,\eta) &= \Big\vert \frac{|B_l|}{|B_{\eps N}|} \sum_{j\in J} \big( \frac{1}{|B_l|} \sum_{y\in B_{l}^j} g(\tau_y\eta) - \<g\> (M_{B_{\eps N}(x)}(\eta)) \big) \Big\vert \notag\\
											 &\leq \frac{|B_l|}{|B_{\eps N}|} \sum_{j\in J} \Big\vert \frac{1}{|B_l|} \sum_{y\in B_{l}^j} g(\tau_y\eta) - \<g\> (M_{B_l^j}(\eta)) \Big\vert \notag\\ 
											 &\quad + \frac{|B_l|}{|B_{\eps N}|} \sum_{j\in J} \Big\vert \<g\> (M_{B_l^j}(\eta)) - \<g\> (M_{B_{\eps N}(x)}(\eta)) \Big\vert. 
	\end{align}
	The latter summand in \eqref{eq:estimate_V_eps_N} may be estimated by 
	\begin{multline}
		\|\<g\>'\|_\infty  \frac{|B_l|}{|B_{\eps N}|} \sum_{j\in J} \big\vert M_{B_l^j}(\eta) - M_{B_{\eps N}(x)}(\eta)   \big\vert\\ \leq \|\<g\>'\|_\infty  \frac{|B_l|^2}{|B_{\eps N}|^2} \sum_{j\in J} \sum_{i\in J} \big\vert M_{B_l^j}(\eta) -  M_{B^i_l}(\eta) \big\vert. \label{eq:one_more_estimate} 
	\end{multline}
	Combining \eqref{eq:estimate_V_eps_N}, \eqref{eq:one_more_estimate}, Lemma~\ref{lem:one_block_estimate} and Lemma~\ref{lem:two_block_estimate}, we conclude the statement of Lemma~\ref{lem:replacement}. 
\end{proof}

\end{document}